\pdfoutput=1

\documentclass[11pt]{article}

\usepackage{amsmath,amssymb,amsfonts,amsthm}
\usepackage{graphicx}
\usepackage{url,hyperref}
\usepackage[ruled,vlined,linesnumbered,noend]{algorithm2e}
\usepackage{tcolorbox}
\usepackage{longtable}
\usepackage{tikz}
\usepackage{cleveref}
\usepackage{thm-restate}
\usepackage[margin=1in]{geometry}

\usetikzlibrary{positioning}

\title{Dynamic Dominating Set in Uniformly Sparse Graphs}

\author{Anton Bukov\thanks{Tel Aviv University, Israel. \href{mailto:bukov.anton@gmail.com}{bukov.anton@gmail.com}} \and
Shay Solomon\thanks{Tel Aviv University, Israel. \href{mailto:shayso@tauex.tau.ac.il}{shayso@tauex.tau.ac.il}}}

\date{}

\newcommand{\wts}{\omega}
\newcommand{\dom}{\operatorname{dom}}
\newcommand{\Dom}{\operatorname{Dom}}
\newcommand{\OPT}{OPT}

\newcommand{\sets}{\mathcal{S}}

\newcommand{\univ}{\mathcal{U}}
\newcommand{\eps}{\epsilon}

\newcommand{\lev}{\mathsf{lev}}
\newcommand{\zlev}{\mathsf{zlev}}
\newcommand{\bs}{\mathsf{base}}
\newcommand{\del}{\mathsf{Delete}}
\newcommand{\ins}{\mathsf{Insert}}

\newcommand{\reset}{\mathsf{Rebuild}}
\newcommand{\water}{\mathsf{WaterFilling}}
\newcommand{\forget}{\text{forget}}
\newcommand{\passive}{\text{passive}}
\newcommand{\level}{\text{level}}
\newcommand{\old}{\mathrm{old}}
\newcommand{\sz}[1]{\left|#1\right|}

\newcommand{\ceil}[1]{\left\lceil #1 \right\rceil}
\newcommand{\floor}[1]{\left\lfloor #1 \right\rfloor}

\newcommand{\wat}[1]{\lambda (1 + \epsilon)^{- #1}}
\newcommand{\viewx}[2]{\hat{x}_{#2}(#1)}
\newcommand{\view}[2]{\widehat{\lev}_{#1}(#2)}
\newcommand{\Forg}[1]{F(#1)}
\newcommand{\Ign}[2][]{I_{#1}(#2)}

\newcommand{\bt}{\mathsf{bot}}

\DeclareMathOperator{\poly}{poly}

\newtheorem{theorem}{Theorem}[section]
\newtheorem{lemma}{Lemma}[section]
\newtheorem{definition}{Definition}[section]
\newtheorem{corollary}{Corollary}[section]
\newtheorem{claim}{Claim}[section]
\newtheorem{observation}{Observation}[section]
\newtheorem{invariant}{Invariant}[section]
\newtheorem{invariantintro}[invariant]{Invariant}
\newtheorem{property}{Property}[section]
\newtheorem{step}{Step}[section]
\newtheorem{assumption}{Assumption}[section]
\newtheorem{question}{Question}[section]
\theoremstyle{definition}
\newtheorem{remark}[theorem]{Remark}
\newtheorem{note}[theorem]{Note}

\newenvironment{claimproof}[1][\proofname]{%
  \begin{proof}[#1]%
}{%
  \end{proof}%
}

\crefname{claim}{Claim}{Claims}
\Crefname{claim}{Claim}{Claims}
\crefname{invariant}{Invariant}{Invariants}
\Crefname{invariant}{Invariant}{Invariants}
\crefname{invariantintro}{Invariant}{Invariants}
\Crefname{invariantintro}{Invariant}{Invariants}
\crefname{property}{Property}{Properties}
\Crefname{property}{Property}{Properties}
\crefname{step}{Step}{Steps}
\Crefname{step}{Step}{Steps}
\crefname{assumption}{Assumption}{Assumptions}
\Crefname{assumption}{Assumption}{Assumptions}
\crefname{question}{Question}{Questions}
\Crefname{question}{Question}{Questions}
\crefname{observation}{Observation}{Observations}
\Crefname{observation}{Observation}{Observations}

\begin{document}

\maketitle

\begin{abstract}
  In the dynamic \emph{minimum dominating set (MDS)} problem, the goal is to efficiently maintain an approximate MDS in an $n$-vertex graph with vertex costs in $[1/C,1]$ undergoing edge insertions and deletions.
  In STACS'19 \cite{HIPS19} it was shown that an $O(\log n)$-approximate MDS can be maintained in {\em unweighted graphs} with $O(\Delta \cdot \log n)$ update time, where $\Delta$ is an upper bound on the maximum degree throughout the update sequence, and in STOC'23 \cite{solomon2023dynamic} this was extended to weighted graphs and improves the approximation guarantee to $(1+\eps)\ln \Delta$.

  Is it possible to achieve $\mathrm{poly}(\log n)$ update time without any dependence on $\Delta$, for any nontrivial graph family? This basic question has remained open even in {\bf forests} and even for {\bf unweighted instances}.
  The {\em arboricity} $\alpha=\alpha(G)$ of a graph $G$ is the minimum number of edge-disjoint forests whose union is $G$, and is a standard measure of sparsity. While $\alpha$ is bounded by $\Delta$ in any graph, various real-world graph families exhibit a significant gap between $\alpha$ and $\Delta$.

  In this work, we show that one can maintain an $O(\alpha)$-approximate MDS with update time $O(\alpha \cdot \log (Cn))$, for dynamic graphs whose {\em arboricity} is bounded by $\alpha$ throughout the update sequence. This replaces the dependence on $\Delta$ in prior update bounds with $\alpha$, while also improving the approximation guarantee for bounded-arboricity graphs. In particular, for any graph family of constant arboricity, such as planar graphs, bounded treewidth graphs, and more generally graphs excluding a fixed minor, our algorithm gives an $O(1)$-approximation with $O(\log (Cn))$ update time. To achieve this result, our algorithm departs from prior {\em greedy-based} approaches, relying instead on the {\em primal-dual framework} and new structural insights specific to bounded arboricity graphs.
\end{abstract}

\clearpage
\tableofcontents
\clearpage

\section{Introduction}

This work studies the classic problem of {\em minimum dominating set (MDS)}, where we are given an $n$-vertex $m$-edge graph $G = (V,E)$ in which each vertex $v$ has a cost $c_v \in [\frac{1}{C}, 1]$. A subset of vertices $D \subseteq V$ is called a {\em dominating set} if, for every vertex $v \in V$, either $v \in D$ or $v$ has a neighbor in $D$. The goal is to compute a dominating set of minimum total cost.

This problem is closely related to the {\em minimum set cover (MSC)} problem: given a set system $(\univ, \sets)$, where $\univ$ is a universe of $n$ elements and $\sets$ is a family of $m$ sets, each with cost in $[\frac{1}{C}, 1]$, the goal is to compute a minimum-cost subfamily covering $\univ$. The two problems admit approximation-preserving reductions and have been studied in many cases together.

A classical {\em greedy} algorithm achieves logarithmic approximations for both problems. More precisely, letting $\Delta = \Delta(G)$ denote the maximum degree of $G$ and $D$ the maximum set size, the greedy algorithm gives an $H(\Delta+1)$-approximate MDS and an $H(D)$-approximate MSC, where $\forall k, H(k) := \sum_{i=1}^k \frac{1}{i}$. A {\em primal-dual (PD)} algorithm yields a $(\Delta+1)$-approximate MDS and an $f$-approximate MSC, where $f$ is the {\em frequency} of the set system.

The greedy approach is essentially optimal for both problems, as one cannot achieve better approximations (to within a factor of $1-\epsilon$) unless P = NP \cite{williamson2011design,dinur2014analytical,CC08}.
The PD approach behaves differently: an $(f-\epsilon)$-approximation cannot be achieved in polynomial time under the unique games conjecture \cite{khot2008vertex}, while the $(\Delta+1)$-approximation obtained for MDS by the PD algorithm is exponentially worse than the greedy guarantee.
However, in {\bf uniformly sparse graphs} the situation changes,
and the PD approach can be used to obtain significantly improved approximations for MDS, which is the focus of this work.

The MDS problem has been studied extensively in various settings and computational models. For brevity, the literature survey that follows is restricted to previous work on the MDS problem either for uniformly sparse graphs or in the dynamic setting.

\medskip
\noindent
{\bf Uniformly sparse graphs.~}
A standard notion that captures sparsity in a uniform manner is {\em arboricity}.
A graph $G=(V,E)$ has {\em arboricity} $\alpha=\alpha(G)$ if $\alpha = \max_{S \subseteq V} \left\lceil \frac{m_S}{n_S-1} \right\rceil$, where $m_S$ and $n_S$ denote the number of edges and vertices in the subgraph induced by $S$, respectively.
By the Nash-Williams Theorem~\cite{nash1964}, $\alpha$ is the minimum number of edge-disjoint forests into which the edges of $G$ can be partitioned.
Up to constants, arboricity equals the maximum average degree over all subgraphs and other basic measures, such as {\em degeneracy} and {\em thickness}. While $\alpha(G) \le \Delta(G)$ holds for every graph, the gap between the two parameters can be large, e.g., for the $n$-star, $\alpha=1$ but $\Delta = n-1$. Significant gaps between $\Delta$ and $\alpha$ are common in real-world networks, including social networks, the world wide web graph, and transaction graphs, as well as in stochastic settings such as the preferential attachment model.
The family of bounded-arboricity graphs is broad: even for $\alpha = O(1)$, it includes all graphs excluding a fixed minor, and hence all bounded-treewidth and bounded-genus (including planar) graphs.

In the classic sequential setting, an
$O(\alpha)$-approximation can be computed in near-linear time \cite{BU17,MSW21,sun2021,DGI22}; in particular, Sun \cite{sun2021}
gave a fast PD $(\alpha+1)$-approximation algorithm for weighted graphs of arboricity $\alpha$. This bound is almost tight: one cannot achieve $(\alpha-\eps)$-approximation in polynomial time under the unique games conjecture \cite{BU17}.

In the $\mathsf{CONGEST}$ model of distributed computing,
improving over \cite{LW10,MSW21},
Dory et al.\ \cite{DGI22} presented a deterministic $(2\alpha+1)(1+\epsilon)$-approximation algorithm running in $O(\epsilon^{-1} \log \Delta)$ rounds and a randomized
$\alpha(1 + o(1))$-approximation algorithm running in $O(\alpha \log \Delta)$ rounds.

\medskip
\noindent
{\bf The dynamic setting.~}
We study the MDS problem in the {\em standard dynamic setting}, where the graph undergoes a sequence of edge insertions and deletions on a fixed vertex set, starting from an empty graph. The goal is to maintain an approximate MDS with low update time; we focus on optimizing the {\em amortized} rather than {\em worst-case} update time; for brevity, we will mostly ignore the distinction between these two update time measures.

There are only a few results on the MDS problem in the dynamic setting, and they are obtained by adapting dynamic MSC algorithms. In contrast, the dynamic MSC problem has been studied extensively over the past decade. Most of the known algorithms are PD-based \cite{bhattacharya2015design,gupta2017online,abboud2019dynamic,assadi2021fully,bhattacharya2019new,bhattacharya2021dynamic,bukov2025nearly,solomon2024lossless}, and they provide $(1+\eps)f$-approximate MSC with increasingly better update times,\footnote{For brevity, we will assume in this section that $\epsilon$ is a small {\em constant}.} culminating in deterministic and randomized update times of $O(f \log f)$ and $O(f \cdot \log^*f)$, respectively \cite{bukov2025nearly}. More relevant to this work are the greedy-based algorithms \cite{gupta2017online,solomon2023dynamic,solomon2024lossless},
which provide $(1+\eps)\ln n$-approximate MSC with update time $O(f \log n)$.

By adapting the greedy-based dynamic MSC algorithm of \cite{gupta2017online}, Hjuler et al.\ \cite{HIPS19} showed that an
$O(\log n)$-approximate MDS can be maintained with update time $O(\Delta \log n)$ in unweighted graphs, where $\Delta$ is an upper bound on the maximum degree throughout the update sequence. Solomon and Uzrad \cite{solomon2023dynamic} (see also \cite{solomon2024lossless}) designed refined greedy-based dynamic algorithms for MSC and MDS in weighted instances; for weighted MDS, they achieved $(1+\eps)\ln \Delta$-approximation with update time $O(\Delta \log n)$.
This yields an update time of $\tilde{O}(1)$
in bounded-degree graphs, where we use the notation $\tilde{O}$ to suppress $\mathbf{poly}(\log n)$ terms.
However, the question of whether one can replace the dependence on $\Delta$ by $\alpha$ and obtain $\tilde{O}(\alpha)$ update time has remained elusive, {\bf even for forests} and {\bf unweighted instances}.

\begin{tcolorbox} [width=\linewidth, sharp corners=all, colback=white!95!black]
\begin{question} \label{q1}
Can one maintain an $O(\alpha)$-approximate MDS with $\tilde{O}(\alpha)$ update time?
\end{question}
\end{tcolorbox}

\subsection{Our Contribution}
In this work we prove the following theorem,
which resolves \Cref{q1} in the affirmative.

\begin{theorem}\label{thm:main}
  There is a deterministic algorithm that maintains an $O(\alpha)$-approximate MDS
  with amortized update time $O(\alpha \log (Cn))$ in any dynamic graph of arboricity bounded by $\alpha$.
\end{theorem}

\noindent
{\bf Remarks.}
\begin{itemize}
    \item Our result replaces the dependence on $\Delta$ in the update time $O(\Delta \log n)$ of the dynamic MDS algorithms \cite{gupta2017online,solomon2023dynamic,solomon2024lossless} by a dependence on $\alpha$, but it incurs a logarithmic dependence on the aspect ratio $C$. We remind that $\alpha \le \Delta$ in any graph.
\item Our approach yields an approximation factor that approaches $6\alpha +2$. While we can reduce the leading constant 6 with an extra effort, our approach cannot reduce it all the way to 1 (as done in the static case \cite{sun2021,DGI22}); this remains an intriguing open problem.
\end{itemize}

\subsection{Perspective}
In contrast to MDS, for several fundamental problems, including maximal matching (MM) and maximal independent set (MIS), and approximate maximum matching, minimum vertex cover and maximum independent set, there exist dynamic algorithms
for graphs with arboricity bounded by $\alpha$
with update time
$\poly(\alpha,\log n)$.
At a high level, these results are typically obtained by reducing bounded-arboricity graphs to bounded-degree graphs, using either {\em low outdegree orientations} or {\em bounded-degree sparsifiers}, as we next discuss.

\smallskip
\noindent
{\bf MM and MIS.~}
For MM and MIS, the known reductions \cite{neiman2015simple,he2014orienting,kopelowitz2014orienting,chekuri2024adaptive,onak2020fully} are based on dynamically maintaining a {\em low outdegree orientation}, i.e., an orientation of the edges of the graph in which the maximum outdegree is close to the graph's arboricity $\alpha$. Dynamic algorithms for low outdegree orientations have been well studied
\cite{brodal1999dynamic,he2014orienting,kopelowitz2014orienting,berglin2020simple,chekuri2024adaptive}, and it is long known that an outdegree of $O(\alpha)$ can be maintained with $O(\log n)$ update time \cite{brodal1999dynamic}.
Once each vertex has outdegree $O(\alpha)$, one can simulate the corresponding ``bounded-degree algorithms'' (i.e., the dynamic algorithms that perform efficiently on bounded-degree graphs) in bounded-arboricity graphs, by cleverly restricting attention to out-neighborhoods rather than full neighborhoods,
in order to reduce the $\Delta$-dependencies in the update time bounds to $\alpha$-dependencies.

\smallskip
\noindent
{\bf Approximate problems.~}
More relevant to MDS than (exact) MM and MIS are the problems of {\em approximate} maximum matching, minimum vertex cover, and maximum independent set. In these problems, the key tool (used sometimes implicitly) for the reductions is that of {\em bounded-degree sparsifiers} \cite{PS16,Sol18}, combined with the basic {\em periodic rebuilding} approach \cite{gupta2013fully}. The bounded-degree sparsifiers \cite{PS16,Sol18} rely on simple local rules and can be maintained efficiently in dynamic graphs of bounded arboricity. In the matching sparsifier of~\cite{Sol18}, each vertex marks $\Delta = O(\alpha/\varepsilon)$ arbitrary neighbors, and only edges marked by both endpoints are included in the sparsifier; the resulting subgraph has a maximum degree at most $\Delta$ and it preserves the maximum matching size to within a factor of $1+\epsilon$. Similarly, for $(1+\epsilon)$-approximate vertex cover sparsification, all vertices of degree at least $\Delta = O(\alpha/\varepsilon)$ are included in the solution, and the problem is reduced to the induced subgraph on the vertices of degree $< \Delta$; independent set sparsification is symmetric to vertex cover sparsification.

\smallskip
\noindent
{\bf MDS.~}
The MDS problem appears to be significantly more challenging to dynamize.

For the problems above, the dynamic bounded-degree algorithms are simple. To maintain an MM dynamically, it suffices to scan the neighborhoods of the at most two updated vertices after each edge update, yielding a trivial $O(\Delta)$ update time algorithm.
Using a dynamic low out-degree orientation as a black-box, extending this algorithm to bounded-arboricity graphs is straightforward.
For dynamic MIS, the bounded-degree algorithm is also simple and local, though the extension to bounded-arboricity graphs requires a nontrivial effort.
For the approximate problems mentioned above, the dynamic bounded-degree algorithms are obtained by combining the periodic rebuilding approach with simple structural reductions, and extending them to bounded-arboricity graphs can be done in a straightforward manner using the bounded-degree sparsifiers of \cite{PS16,Sol18}.

For MDS, however, the dynamic bounded-degree algorithms are much more complex than the algorithms for the other problems above.
We focus on the dynamic PD algorithms \cite{bhattacharya2015design,bhattacharya2019new,bhattacharya2021dynamic,solomon2024lossless,bukov2025nearly}, since even in static graphs, all known $O(\alpha)$-approximation algorithms \cite{BU17,MSW21,sun2021,DGI22} are based on linear programming (and especially PD) techniques. (We note that the algorithm of \cite{MSW21} can also be cast as a PD algorithm, although it is not presented as such.)
These dynamic PD algorithms provide $(1+\eps)f$-approximate MSC with update time $O(f \log n)$, and are directly translated to provide $(1+\eps)(\Delta+1)$-approximate MDS with update time $O(\Delta \log n)$. Importantly, our goal is to replace the $\Delta$-dependence in {\em both} the approximation factor and the update time by an $\alpha$-dependence.

These algorithms are organized around maintaining, or periodically restoring, the underlying {\em complementary slackness conditions}. That is, some of the algorithms \cite{bhattacharya2015design} maintain these conditions locally at all times, while others \cite{bhattacharya2019new,bhattacharya2021dynamic,solomon2024lossless,bukov2025nearly} allow the constraints to be violated as long as the overall approximation is in check, and they fix the violated constraints periodically, as part of {\em global rebuild} procedures. Both the local fixing rules and the global rebuild procedures rely heavily on scanning entire ``neighborhoods'' of vertices or elements (in the set cover problem),
in order to maintain or restore the complementary slackness conditions.
Alas, such scans incur at least a linear dependence on $\Delta$ in the update time.

To translate the dependence on $\Delta$ in the update time to an $\alpha$-dependence,
scanning entire neighborhoods is no longer feasible. Instead,
to achieve an update time proportional to $\alpha$ rather than $\Delta$, one must work with only an $O(\alpha)$-sized subset of each neighborhood. This creates an immediate basic obstacle: in the MDS problem, even a single neighbor can be crucial (e.g., a vertex that uniquely dominates a large portion of the graph), and there is no known simple local rule, analogous to those used for matching, vertex cover or independent set, that allows one to safely ignore most neighbors.

This basic obstacle creates several highly nontrivial algorithmic and analytical challenges.
These challenges become even more significant, since our goal is also to reduce the approximation from $O(\Delta)$ to $O(\alpha)$ while maintaining fast update time.
We discuss these challenges and our approach in more detail in \cref{sec:previous} and \cref{sec:technicalov}.

\subsection{Organization}

In \cref{sec:previous} we review the primal-dual framework for MDS, focusing on prior dynamic algorithms and static algorithms for bounded-arboricity graphs.

In \cref{sec:technicalov} we give a technical overview, highlighting the main challenges in replacing the $\Delta$-dependence by an $\alpha$-dependence and outlining our key ideas, including lazy maintenance of packing values, sparsification of the rebuild procedure, and the forgetting mechanism.

In \cref{sec:algo} we provide the complete description of our dynamic algorithm with all the technical details, and in \cref{sec:analysis} we provide the amortized runtime analysis. The glossary of notation is given in \cref{sec:glossary}. Finally, in \cref{sec:conclusion} we provide the conclusion and future directions.

\section{Preliminaries and Previous approaches} \label{sec:previous}
\subsection{The Primal-Dual framework for MDS}
We denote by $N[v]$ the {\em closed neighborhood} of $v$. For a set $S \subseteq V$, we write $N[S] = \bigcup_{v \in S} N[v]$ for the set of vertices dominated by $S$. Let $L = \left\lceil \log_{1+\epsilon} (Cn) \right\rceil + 1$. Let $\tau_v$ be the cost of a cheapest neighbor of $v$, i.e.\ $\tau_v = \min_{u \in N[v]} c_u$. Let $\epsilon \in (0, 1/3)$.

Our algorithm is a direct adaptation of the primal-dual framework of \cite{bhattacharya2015design,bhattacharya2019new,bhattacharya2021dynamic} from MSC to MDS, and we present the definitions only for MDS for brevity.

For the MDS problem, each vertex $v$ is assigned a {\em packing value} $x_v \ge 0$, and the weight of a vertex $v$ is $\wts(v) = \sum_{u \in N[v]} x_u$. A packing is feasible if $\wts(v) \le c_v$ for every vertex $v$. Let $\OPT$ denote the minimum cost of any dominating set. For any $S \subseteq V$, let $c(S) = \sum_{v \in S} c_v$.

The following lemma is the MDS analogue of the corresponding lemma for MSC, used in all previous dynamic primal-dual algorithms \cite{bhattacharya2015design,bhattacharya2019new,bhattacharya2021dynamic,bukov2025nearly}.

\begin{lemma}\label{lm:packing}
  Consider any feasible fractional packing $\{x_v\}_{v \in V}$ for MDS. Then $\sum_{v \in V} x_v \le \OPT$. In addition, if there is a dominating set $D \subseteq V$ and $\rho \ge 1$ such that
    $\wts(v) \ge c_v / \rho$ for all $v \in D$,
  then $c(D) \le \rho(\Delta+1)\OPT$.
\end{lemma}
\begin{proof}
  Let $D^{*}$ be an optimal dominating set. Then
  \[
    \sum_{v \in V} x_v \le \sum_{u \in D^{*}} \sum_{v \in N[u]} x_v = \sum_{u \in D^{*}} \wts(u) \le \sum_{u \in D^{*}} c_u = \OPT.
  \]

  For the second claim,
  \[
    c(D) \le \rho \sum_{v \in D} \wts(v)
    \le \rho(\Delta+1)\sum_{u \in V} x_u
    \le \rho(\Delta+1)\OPT,
  \]
  since each $x_u$ appears in at most $\Delta+1$ closed neighborhoods.
\end{proof}

A fractional packing and a dominating set satisfying the conditions of \cref{lm:packing} can be obtained by the classic {\em water-filling} primal-dual algorithm. For efficiency, this algorithm is {\em discretized} by introducing a {\em hierarchical partition} of vertices into levels. Specifically, each vertex $v$ is assigned a \emph{dominating level} $\dom(v) \in [L]$, and the \emph{dominated level} of $v$ is defined as $\Dom(v) = \max_{u \in N[v]} \dom(u)$; note that $\Dom(v) \ge \dom(v)$. The packing value of $v$ is then given by $x_v = (1+\epsilon)^{-\Dom(v)}$. The water-filling algorithm determines the levels of vertices as follows. It starts by assigning $\dom(v) = L$ for all vertices $v$, and it gradually decreases the dominating levels, stopping for a vertex $v$ once $\wts(v) > c_v/(1+\epsilon)$ or $\dom(v)=\bs(v)$, where $\bs(v) = \floor{\log_{1+\epsilon}\frac{1}{c_v}}$ is the \emph{base level} of $v$. Define by $\bt(v) = \floor{\log_{1+\epsilon}\frac{1}{\tau_v}}$ the \emph{bottom level} of $v$; note that $\bt(v) = \max_{u \in N[v]} \bs(u)$.

Let $T \subseteq V$ be the collection of vertices $v$ satisfying $\wts(v) > c_v/(1+\epsilon)$. We call vertices satisfying this condition \emph{tight}; other vertices are called \emph{slack}. We note that every vertex $v$ with $\Dom(v)=\bt(v)$ belongs to $T$, since $x_v \ge \tau_v$. Thus, $T$ is a dominating set that satisfies the conditions of \cref{lm:packing}, and so $T$ is a $(1+\epsilon)(\Delta+1)$-approximate MDS.

\subsection{Previous Dynamic Primal-Dual Algorithms}
Previous dynamic primal-dual algorithms for the MDS problem are obtained by dynamizing the $\water$ subroutine described above. In particular, the dynamic algorithms of \cite{bhattacharya2019new,bhattacharya2021dynamic,bukov2025nearly} all use, as a basic rebuilding primitive, the following MDS adaptation of the water-filling algorithm from \cite{bhattacharya2019new}.

\begin{lemma}[From \cite{bhattacharya2019new}, adapted for MDS]\label{lm:water-filling}
  There exists a subroutine $\water$ with the following properties.\footnote{In \cite{bhattacharya2019new}, this subroutine is called $\textrm{FixLevel}$.} The subroutine $\water(k, D', V')$ takes as input two collections of vertices $D'$ and $V'$ such that $\dom(v)=k$ for every $v \in D'$ and $\Dom(u)=k$ for every $u \in V'$. Moreover, each vertex $v \in D'$ satisfies $\wts(v) \le c_v$. The subroutine updates the dominating levels of vertices in $D'$ and the dominated levels of vertices in $V'$ so that: (1) $\wts(v) \le c_v$ for every $v \in D'$, and (2) if $\dom(v) > \bs(v)$, then $\wts(v) > c_v/(1+\epsilon)$ for every $v \in D'$. The subroutine runs in time $O\!\left(k + |D'| + \sum_{u \in V'} |N[u]|\right)$.
\end{lemma}

The algorithm of \cite{bhattacharya2019new} is based on a global rebuild strategy: rather than maintaining a solution that always satisfies the conditions of \cref{lm:packing}, it keeps the current solution as long as its cost remains sufficiently close to $(\Delta+1)\sum_{v \in V} x_v$. When the algorithm determines that the solution might deviate significantly from this target, it chooses an appropriate level $k$, lets $D_{\le k}$ denote the set of vertices $v$ with $\dom(v)\le k$ and $V_{\le k}$ denote the set of vertices $v$ with $\Dom(v)\le k$, raises every vertex in $D_{\le k}$ to dominating level $k$, raises every vertex in $V_{\le k}$ to level $k$, and then invokes the subroutine from \cref{lm:water-filling}. The algorithms of \cite{bhattacharya2021dynamic,bukov2025nearly} follow essentially the same global rebuild paradigm, still using \cref{lm:water-filling} as the core rebuilding step, while additionally applying local fixing rules during updates to improve the update-time guarantees.

This rebuild paradigm is exactly what causes the linear dependence on $\Delta$ in the update time of previous dynamic MDS algorithms. Indeed, when \cref{lm:water-filling} is applied as described above, its running time is proportional to $k + |D_{\le k}| + \sum_{u \in V_{\le k}}(\deg(u)+1)$, since the subroutine must inspect the entire neighborhood of each processed vertex. Thus, rebuilding a prefix of the level hierarchy incurs a cost proportional to the sum of the degrees of the involved vertices. In bounded-degree graphs, this yields efficient update times, but in general graphs it necessarily introduces a linear dependence on $\Delta$.

The same issue arises in the local fixing rules. More generally, all previous dynamic MDS algorithms rely on scanning the full neighborhoods of vertices, both in rebuild subroutines and in local updates. Thus, their update-time bounds inherently depend linearly on $\Delta$.

\subsection{Static Primal-Dual Algorithms in Bounded Arboricity Graphs}
Static $O(\alpha)$-approximation algorithms for MDS on graphs of arboricity at most $\alpha$ were obtained in a sequence of works \cite{BU17,MSW21,sun2021,DGI22}. Our algorithm builds on the distributed algorithm and analysis of Dory et al.\ \cite{DGI22}, which exploit low-outdegree orientations of bounded-arboricity graphs.

The following lemma and claim are direct consequences of the work of Dory et al.\ \cite{DGI22}.

\begin{lemma}\label{lm:1}
  There exists a static algorithm that computes a feasible packing $\{x_v\}_{v \in V}$ such that, with $H$ defined as the set of vertices not dominated by $T$ and $\lambda = \frac{1}{(1+\epsilon)^2(2\alpha+1)}$, the following hold:
  (1) $c(T) \le (1+\epsilon)(2\alpha+1)\sum_{v \in V \setminus H} x_v$; (2) $x_v \ge \lambda\tau_v$ for every $v \in H$; (3)
$x_v \le (1+\epsilon)\lambda\tau_v$ for every $v \in V$.
\end{lemma}

The proof of this lemma uses the following observation.

\begin{observation}\label{obs:orientation-intro}
  The edges of any graph $G$ of arboricity at most $\alpha$ can be oriented so that every vertex has out-degree at most $\alpha$.
\end{observation}
\begin{proof}
  Since $G$ has arboricity at most $\alpha$, its edges can be partitioned into at most $\alpha$ edge-disjoint forests. The edges of a forest can be oriented so that each vertex has at most one outgoing edge. This can be done by fixing a root vertex in each tree of the forest and orienting every edge toward this root. Thus, applying such an orientation to each forest in the partition yields the required orientation.
\end{proof}
\begin{proof}[Proof of \Cref{lm:1}]
  We run the water-filling algorithm with the following modification. We \emph{scale down the packing values by a factor of $\lambda$}, redefining the packing value of $v$ to be $x_v = \lambda (1+\epsilon)^{-\Dom(v)}$. Therefore, the resulting packing satisfies $x_v \le (1+\epsilon)\lambda\tau_v$. Every vertex with $\dom(v) > \bs(v)$ in the resulting packing still has $\wts(v) > c_v / (1+\epsilon)$, and so $v$ is tight. Thus, every vertex $v$ with $\Dom(v) > \bt(v)$ is dominated by $T$. However, we cannot claim that vertices with $\Dom(v) = \bt(v)$ are dominated by $T$, since we scaled down the packing values. But for these vertices we have $x_v \ge \lambda\tau_v$.

  Consider the orientation from \Cref{obs:orientation-intro}. Then
  \begin{equation}\label{eq:5}
    \wts(v)=\sum_{u\in N^{in}[v]}x_u+\sum_{u\in N^{out}(v)}x_u
    \le\sum_{u\in N^{in}[v]}x_u+(1+\epsilon)\alpha\lambda c_v.
  \end{equation}
  Summing over $v\in T$ and swapping the order of summation gives
  \begin{equation}\label{eq:in-sum-intro}
    \sum_{v\in T}\sum_{u\in N^{in}[v]}x_u
    =\sum_{u\in V \setminus H}x_u\cdot|\{v\in T : v\in N^{out}[u]\}|
    \le(\alpha+1)\sum_{v\in V \setminus H}x_v.
  \end{equation}
  For a tight vertex $v\in T$, we have $c_v < (1+\epsilon)\wts(v)$.
  Summing over $T$ and using \cref{eq:5,eq:in-sum-intro} yields
  \[
    c(T)\le\sum_{v\in T}(1+\epsilon)\wts(v)
    \le (1+\epsilon)\sum_{v \in T}\Bigl(\sum_{u \in N^{in}[v]}x_u + (1+\epsilon)\alpha\lambda c_v\Bigr) \le (1+\epsilon)(\alpha+1)\sum_{v\in V \setminus H}x_v + \frac{\alpha}{2\alpha + 1}\sum_{v \in T} c_v,
  \]
  hence
    $c(T)\le (1+\epsilon)(2\alpha+1)\sum_{v\in V \setminus H}x_v$.
\end{proof}

Let $H'$ be the set of cheapest neighbors of the vertices in $H$ from \cref{lm:1} (for each vertex, we take a single neighbor). Then we make the following claim.

\begin{claim}\label{cl:bound-intro}
  The set $T \cup H'$ from \cref{lm:1} forms a $(1+3\epsilon)(2\alpha+1)$-approximate MDS.
\end{claim}
\begin{claimproof}
  Since $x_v \ge \lambda\tau_v = \frac{\tau_v}{(1+\epsilon)^2(2\alpha+1)}$ for $v \in H$, we have
  \[
    c(H') = \sum_{v \in H} \tau_v \le (1+\epsilon)^2(2\alpha+1)\sum_{v \in H} x_v.
  \]
  Therefore, $c(T \cup H') \le (1+3\epsilon)(2\alpha+1)\sum_{v \in V} x_v \le (1+3\epsilon)(2\alpha+1)\OPT$.
\end{claimproof}

The algorithm from \cite{bhattacharya2019new} waits until the cost of the solution deviates significantly from $(\Delta+1) \sum_{v \in V} x_v$. However, if we aim for an $O(\alpha)$ approximation ratio, we need it to be close to $\alpha \sum_{v \in V} x_v$. Thus, we might have to rebuild the solution much more frequently, and this requires us to reduce the runtime cost of the rebuild subroutine significantly, from $O(\Delta)$ per vertex to $O(\alpha)$.

Another issue in applying the algorithm from \cite{bhattacharya2019new} is keeping track of the vertices that are {\em undominated} by $T$. However, since our goal is an $O(\alpha)$ approximation, we can simply define $H$ to consist of {\em all} vertices $v$ such that $x_v \ge \lambda\tau_v$ (which is a superset of the original $H$ from \cref{lm:1}). This may cause $H'$ and $T$ to intersect, but the approximation ratio would increase by at most a factor of $2$.

\section{Technical Overview} \label{sec:technicalov}
\subsection{Overview of the Dynamic Algorithm}
At a high level, our algorithm combines the static distributed algorithm for arboricity-$\alpha$ graphs from \cite{DGI22} with the dynamic primal-dual algorithms of \cite{bhattacharya2019new}.
The main obstacle is that combining these approaches directly leads to an update time that depends at least linearly on $\Delta$. The static distributed bounded-arboricity analysis yields the desired $O(\alpha)$ approximation, but the known dynamic primal-dual algorithms rely crucially on scanning full neighborhoods, which costs $\Omega(\Delta)$ time. Our task is therefore to preserve the primal-dual structure while replacing all essential $\Delta$-dependent operations by $\alpha$-dependent ones. To achieve $O(\alpha)$ approximation as in the distributed algorithm, additional constraints must be enforced, beyond the usual constraints of the primal-dual dynamic algorithms.
Maintaining these additional constraints is essential for achieving an $O(\alpha)$-approximation, but it makes it significantly more difficult to obtain fast update time. In the following subsections, we summarize the main obstacles and our approach to overcoming them.
We begin by explaining how to avoid scanning entire neighborhoods during updates, and then describe how to sparsify the rebuild subroutine to obtain an $\alpha$-dependent update time.

We use the framework of~\cite{bhattacharya2019new}, which distinguishes between \emph{active} and \emph{passive} vertices. Each vertex is assigned a \emph{level} $\lev(v) \ge \Dom(v)$, and the packing value is $x_v = \lambda(1+\epsilon)^{-\lev(v)}$. We use $\lambda = \frac{1}{(1+\epsilon)^2(3\alpha+1)}$, in contrast to the definition in \cref{lm:1}. This modified choice of $\lambda$ enables faster update time at the cost of increasing the approximation ratio by a constant factor, via the mechanism of ``forgetting'' vertices, explained later. For an active vertex, we maintain its dominated level $\Dom(v)$ exactly in $\lev(v)$. For a passive vertex $v$, we have $\lev(v) > \Dom(v)$.

Passive vertices introduce an additional logarithmic dependency in the update time.
This is because each passive vertex may need to be processed multiple times before becoming active. The key point is that whenever a passive vertex is processed by the rebuild subroutine, the gap $\lev(v)-\zlev(v)$ decreases (we discuss the lazy level $\zlev(v)$ later; intuitively it acts like a lazy version of $\Dom(v)$). Consequently, unless an incident edge update occurs first, any passive vertex can participate in at most $O(\log_{1+\epsilon}(Cn))$ rebuilds before it becomes active.

To simplify this overview, for the rest of this section, we assume that there are no passive vertices (unless stated otherwise). This allows us to convey the main ideas more effectively. However, dealing with passive vertices is a highly nontrivial challenge and we defer this to the full description provided in \cref{sec:algo}. Thus, throughout this overview we assume $\lev(v) = \Dom(v)$.

Following the static algorithm from \cref{lm:1}, we maintain the following invariant.
\begin{invariantintro}[\cref{inv:tight} in \cref{sec:invariants}]\label{inv:tight-overview}
  Every vertex $v$ with $\dom(v) > \bs(v)$ is tight.
\end{invariantintro}
The notion of tightness is defined later in \cref{def:tightness-intro}. Intuitively, a vertex $v$ is tight if its weight is close to $c_v$. Let $T$ be the set of all tight vertices, and $H$ be all vertices with $\lev(v) = \bt(v)$; the definition of $H'$ remains the same as after \cref{lm:1}. The dominating set our algorithm maintains is $T \cup H'$ (we defer the proof that this is indeed a dominating set to \cref{sec:invariants}). We show by \cref{lm:approximation-guarantee-overview} this is an $O(\alpha)$-approximate MDS.

Due to scaling down the packing values by a factor of $\lambda$ and \Cref{inv:tight-overview} we have the following bound on the packing values $x_v \le (1+\epsilon)\lambda \tau_v$. This is different from the previous primal-dual dynamic algorithms, where the bound was $x_v \le (1+\epsilon)\tau_v$ (or simply $x_v \le 1$). This bound on the packing values plays a crucial role in achieving the required bounds for both the approximation factor and the update time; we explain it in the following sections.

\subsubsection{Handling Deletions without Full Neighborhood Scans}
In the standard framework, when an edge disappears, the dominated level $\Dom(v)$ of a vertex $v$ may decrease, and restoring the exact value of $\Dom(v)$ may again require scanning the entire neighborhood $N[v]$.

We handle that issue by keeping a \emph{lazy level} $\zlev(v)\le \Dom(v)$ and not tracking $\Dom(v)$ exactly for passive vertices, in the spirit of \cite{bukov2025nearly}. Thus, whenever an edge is deleted, instead of scanning the neighborhood to compute $\Dom(v)$, we reset $\zlev(v)$ to 0, which enables us to handle deletion in constant time. Later, when $v$ is processed in a rebuild subroutine, either its lazy level is increased (and so gets closer to $\Dom(v)$), or $v$ becomes active. Thus, intuitively, we defer scanning the neighborhood for $v$ to global rebuilds.

We note that in our case, unlike \cite{bukov2025nearly}, $\zlev(v)$ cannot be any level $\le \Dom(v)$, since not every vertex with $\dom(v) > 0$ is taken to the solution. This issue is handled by \cref{inv:zlev} in \cref{sec:invariants}, which guarantees that at least one neighbor is tight and thus is taken to the solution.

\subsubsection{Handling Insertions without Full Neighborhood Scans}
In previous dynamic primal-dual algorithms, when $\Dom(v)$ increases due to an edge insertion, one has to explicitly update the weights of all vertices in $N[v]$.
This is already problematic after a single edge insertion. Indeed, if inserting an edge causes $\Dom(v)$ to increase, then the packing value $x_v$ decreases, and updating the weights $\wts(u)$ of all vertices $u \in N[v]$ may require scanning the entire neighborhood of $v$, leading to $\Omega(\Delta)$ update time.

To avoid this, we maintain weights approximately using a lazy update scheme. One source of error is the \emph{deviation}, denoted by $\delta$,
which captures the discrepancy between the true weights and the maintained lazy weights due to outdated information about neighbors' levels. Another is \emph{dead weight} $\phi(v)$, which we add after deletions to temporarily compensate for lost weight; let $\phi=\sum_{v\in V}\phi(v)$.

Throughout the algorithm we bound the total deviation by \cref{inv:error-bound-overview} provided later.
As long as this bound holds, the approximation degrades by only a factor of $1+O(\epsilon)$. This is formally captured by the following \cref{lm:approximation-guarantee-overview}, the proof of which we defer to \cref{sec:approximation}.

First we explain the concept of viewed levels and how they allow us to avoid scanning the neighborhoods on edge insertions, and then show how this relates to $\delta$ and the resulting approximation guarantee.

In the algorithm, we do not maintain the true weight $\wts(v)=\sum_{u\in N[v]} x_u$ explicitly. Instead, for every vertex $v$ and every neighbor $u\in N[v] \setminus F(v)$, we maintain a \emph{viewed level} $\view{v}{u}$, which represents the level at which $v$ currently views $u$. The set $F(v)$ is the set of forgotten neighbors, which we define later in \cref{sec:rebuild-intro}. The corresponding \emph{viewed value} is then defined by $\viewx{v}{u}=\wat{\view{v}{u}}$. Thus each vertex $v$ stores a \emph{viewed weight} $\hat{\wts}(v)=\sum_{u\in N[v] \setminus F(v)} \viewx{v}{u}$, where always $\viewx{v}{u}\ge x_u$. This inequality is enforced by maintaining the following invariant.
\begin{invariantintro}[\cref{inv:view} in \cref{sec:invariants}]
  For every vertex $v$ and every $u \in N[v] \setminus I(v)$, we have $\dom(u) \le \view{u}{v} \le \lev(v)$.
\end{invariantintro}
We do not maintain this bound for a set of ignoring vertices $I(v)$,
since for them it holds trivially; we explain this later when we
introduce forgotten vertices in \cref{sec:rebuild-intro}. In
particular, this implies $\wts(v) \le \hat{\wts}(v) + |F(v)| \cdot
(1+\epsilon)\lambda c_v$, and hence if $\hat{\wts}(v) \le c_v - |F(v)| \cdot (1+\epsilon)\lambda c_v$ for every vertex $v$, then the packing is feasible. The last part is enforced by the following invariant.
\begin{invariantintro}[\cref{inv:weight-bound} in \cref{sec:invariants}]\label{inv:weight-bound-overview}
  For every vertex $v$ we have $\hat{\wts}(v) \le c_v - |F(v)| \cdot (1+\epsilon)\lambda c_v$.
\end{invariantintro}

Next, let us explain how we maintain the bound $\view{u}{v} \le \lev(v)$. The key point is that during insertions, we may only increase $\lev(v)$, while deletions leave $\lev(v)$ unchanged; $\lev(v)$ may decrease only during rebuilds, but whenever we decrease $\lev(v)$ during rebuild, this is when we update all $\view{u}{v}$ to the actual value $\lev(v)$. Therefore, upon insertions and deletions we can safely let all neighbors $u \in N[v] \setminus I(v)$ keep their old $\view{u}{v}$.

To show how it is related to $\delta$, let us define \emph{deviation for a vertex} as $\delta(v) = \max_{u \in N[v] \setminus I(v)} \viewx{u}{v} - x_v$. Thus, whenever the true value $x_v$ drops due to an insertion, $\delta(v)$ increases. Then $\delta$ is defined as $\sum_{v \in V} \delta(v)$. We note that the actual $\delta(v)$ we use in the algorithm is defined a bit differently to simplify the amortized analysis, but these definitions differ up to a constant factor. Whenever $\lev(v)$ increases due to an insertion, the deviation $\delta(v)$ increases. Conversely, rebuild subroutines reduce deviation by restoring consistency between viewed levels and true levels.

As deviation and dead weight accumulate, the approximation guarantee deteriorates. To control this, we maintain a global invariant that bounds their total contribution. Recall that $T$ is the set of tight vertices.

\begin{invariantintro}[\cref{inv:error-bound} in \cref{sec:invariants}]
\label{inv:error-bound-overview}
$
  (\alpha+1)\delta + \phi \le \epsilon\Bigl(\xi c(T) + (\alpha+1)\sum_{v\in V} x_v\Bigr)
$, where $\xi = \frac{\alpha+1}{(1+\epsilon)(3\alpha+1)}$.
\end{invariantintro}

In \cref{lm:approximation-guarantee-overview} we show that when the invariants are maintained the approximation factor stays close to $6\alpha+2$. Conversely, whenever \cref{inv:error-bound-overview} is violated, we claim that we accumulated enough credits to afford a rebuild subroutine, which reduces $\delta$ and $\phi$ (the full analysis is provided in \cref{sec:analysis}).

\begin{lemma}\label{lm:approximation-guarantee-overview}
  If all the invariants hold, then $c(T) + c(H') \le (1 + 5\epsilon)\cdot 2(3\alpha + 1)\cdot\OPT$.
\end{lemma}
We defer the proof of this lemma to \cref{sec:approximation}.
The proof is similar to \cref{lm:1}, but also accounts for deviation and dead weight, and forgotten neighbors, which we define later.

\subsubsection{Improving the Runtime of Rebuild Via Sparsification}\label{sec:rebuild-intro}
The main bottleneck in previous dynamic primal-dual algorithms is the rebuild step. The algorithm chooses a level $k$ and rebuilds all vertices at level $k$ or below, using the $\water$ subroutine from \cref{lm:water-filling} as the main ingredient. In bounded-degree graphs this is acceptable, because scanning all relevant neighborhoods costs only $O(\Delta)$ per processed vertex. In our setting, however, we cannot afford to inspect every neighbor of every processed vertex.

\paragraph*{Main idea: $\water$ sparsification}
Instead of applying the $\water$ subroutine on the entire subgraph induced by $V_{\le k}$, which may have size $\Theta(|V_{\le k}|\Delta)$, we construct a sparsified subgraph of size roughly $O(|V_{\le k}|\alpha)$ and run $\water$ only on this subgraph. We show that this sparsification preserves the approximation guarantee up to a constant factor.

To explain how we handle this, we will provide a simplified rebuild subroutine that works under the assumption that all vertices are active (recall that we assumed this in previous sections as well), since passive vertices introduce additional challenges. An additional simplifying assumption is that $F(v) \subseteq V_{\le k}$. In the actual algorithm, this might not be the case. We handle that by a more intricate potential analysis that allows us to charge previous rebuilds part of the costs. Then we provide the amortized update time analysis of this simplified subroutine. The full subroutine and its analysis are provided in \cref{sec:rebuild,sec:rebuild-analysis}.

\paragraph*{Reducing the degree in Rebuild}
Consider a rebuild at level $k$. A natural obstacle is that we can afford to scan or process only $O(\alpha)$ (ignoring the dependency on $\epsilon$) relevant neighbors (in $D_{\le k} \setminus V_{\le k}$) per vertex in $V_{\le k}$, but the number of relevant neighbors per vertex in $V_{\le k}$ could be $\Delta$ rather than $O(\alpha)$.
The key structural point is that there cannot be too many vertices in $D_{\le k} \setminus V_{\le k}$ that have many neighbors inside $V_{\le k}$.

\begin{lemma}\label{lm:high-deg}
  Let $B$ denote the set of vertices $u \in D_{\le k} \setminus V_{\le k}$ such that $|N[v] \cap V_{\le k}| > (1+\epsilon)\alpha$. Then $|B|\le \frac{1}{\epsilon}|V_{\le k}|$.
\end{lemma}
\begin{proof}
  Every vertex of $B$ contributes $> (1+\epsilon)\alpha$ edges to the bipartite graph induced by $B$ and $V_{\le k}$, thus this graph contains $> (1+\epsilon)\alpha|B|$ edges. On the other hand, since the graph has arboricity $\le \alpha$, every induced subgraph on $|B|+|V_{\le k}|$ vertices has at most $\alpha(|B|+|V_{\le k}|)$ edges. Therefore, $(1+\epsilon)\alpha|B|<\alpha(|B|+|V_{\le k}|)$, which implies $\epsilon |B|<|V_{\le k}|$, as claimed.
\end{proof}

Thus, all vertices in $D_{\le k} \setminus V_{\le k}$ that have more than $(1+\epsilon)\alpha$ neighbors in $V_{\le k}$ form only a small exceptional set of size $O(|V_{\le k}|/\epsilon)$. Further, we will show that it suffices to focus only on edges between $V_{\le k}$ and this exceptional set. Hence, the total number of such edges would be $O\!\left(\frac{\alpha}{\epsilon}|V_{\le k}|\right)$; equivalently, the average number of such relevant neighbors per vertex of $V_{\le k}$ is $O(\alpha/\epsilon)$. In that case, the rebuilding step, and in particular the call to the $\water$ subroutine from \cref{lm:water-filling} (which deliberately ignores neighbors of vertices in $V_{\le k}$ that lie outside of $V_{\le k} \cup B$), would run in time $O\!\left(k+\frac{\alpha}{\epsilon}\cdot |V_{\le k}|\right)$, instead of $O(k+\Delta\cdot |V_{\le k}|)$. We will later reason that these neighbors can be ignored, and this will have only a small effect on the approximation factor; but it does increase the leading constant in the approx factor.

\paragraph*{Forgetting vertices}
Now we explain the mechanism that allows us to exclude vertices with only a few neighbors in $V_{\le k}$ from the rebuilding process by ``forgetting'' them. We refer to this process as \emph{forgetting} neighbors. The key point is that such vertices contribute only a small amount to the weights, and therefore their effect can be safely approximated.

A key property (by \cref{lm:1}) is that the packing values satisfy $x_v \le (1+\epsilon)\lambda\tau_v$.
Thus the total contribution of any $(1+\epsilon)\alpha$ neighbors to the weight of $v$ is at most $(1+\epsilon)^2\alpha\lambda\tau_v \le \frac{\alpha}{3\alpha+1}c_v$.

Consider a vertex $v$ with $\le (1+\epsilon)\alpha$ neighbors in $V_{\le k}$, and let $\wts'(v)$ be the weight of $v$ without these neighbors. If $\wts'(v) \le \frac{2\alpha+1}{3\alpha+1}c_v$, then even after adding back the contribution of these neighbors, the total weight of $v$ remains at most $c_v$. Therefore, from the perspective of the water-filling process, these neighbors can be ignored without violating feasibility. To make use of this, we relax the notion of tightness. With this relaxed definition, the approximation factor increases only by a constant factor from the original factor of $(1+\epsilon)(2\alpha+1)$ of \cite{DGI22}.

\begin{definition}\label{def:tightness-intro}
  A vertex is called \emph{tight} if $\hat{\wts}(v) + \phi(v) > \frac{2\alpha+1}{(1+\epsilon)(3\alpha+1)}c_v$; otherwise, it is called \emph{slack}.
\end{definition}

The gap between $\frac{2\alpha+1}{(1+\epsilon)(3\alpha+1)}c_v$ and $c_v$ is what allows a slack vertex to forget a bounded number of neighbors without violating feasibility.

To implement this idea, each vertex $v$ maintains a set $F(v)$ of \emph{forgotten} neighbors. If $u \in F(v)$, then $u$ stops contributing to $\hat{\wts}(v)$. Forgotten neighbors do not create outdated copies, and therefore they do not contribute to the deviation. We bound the number of forgotten neighbors with the following invariant.

\begin{invariantintro}[\cref{inv:forgotten-bound} in \cref{sec:invariants}]\label{inv:forgotten-bound-intro}
  For every vertex $v$, we have $|F(v)| \le (1+\epsilon)\alpha$.
\end{invariantintro}

Symmetrically, from the perspective of $u$, if $u \in F(v)$, then $v$ is called an \emph{ignoring} neighbor of $u$. We denote by $I(u)$ the set of neighbors that currently ignore $u$.

Thus, with our relaxed definition of tightness, a slack vertex can safely forget up to $(1+\epsilon)\alpha$ neighbors without violating \cref{inv:weight-bound-overview}. This is the key step that allows us to eliminate most low-impact edges from the subgraph on which the rebuild subroutine operates.

\paragraph*{A simplified rebuild subroutine}
The goal of the rebuild subroutine is to restore \cref{inv:error-bound-overview} by cleaning up all deviation and dead weight up to level $k$, and then to restore \cref{inv:tight-overview} by running the $\water$ subroutine. In our simplified subroutine, we omit the details of how this cleanup is performed, and simply assume that it can be done in constant time. In the actual algorithm, it is done via forgetting some vertices outside $V_{\le k}$, and thus requires a more careful analysis. The full details of the cleanup and the complete rebuild subroutine are provided in \cref{sec:rebuild,sec:rebuild-analysis}.

The pseudocode of our simplified rebuild subroutine is provided in \cref{alg:rebuild-simplified}. At a high level, the rebuild subroutine operates as follows. Let $T'$ be the vertices in $D_{\le k}$ that were tight at the beginning of the rebuild subroutine, and let $V' = V_{\le k}$. First, the subroutine cleans up deviation and dead weight up to level $k$. Then, for every vertex $v \in V_{\le k}$, it attempts to forget $v$ for all $u \in N[v] \setminus I(v)$.
This is implemented using auxiliary sets of ``need-to-be-forgotten'' vertices $F'(v)$, which track which vertices request to be forgotten. More precisely, for each vertex $u$, the set $F'(u)$ will store all vertices $v$ that request $u$ to forget them.
As a result, some vertices $u$ might accumulate more than $(1+\epsilon)\alpha - |F(v)|$ requests in $F'(u)$. Let $S$ be the set of such vertices. The vertices in $S$ correspond to the high-degree vertices from \cref{lm:high-deg}, namely those that cannot safely forget all their incident requests.
For these vertices, and the vertices from $T'$, we reset their forgotten neighbor sets by setting $F(v) = \emptyset$. After that, we invoke the $\water$ subroutine on the vertices in $V'$, considering only neighbors from $N[v] \setminus I(v)$ for each $v \in V'$. Due to the forgetting step, all such relevant neighbors belong to $S \cup T'$, and thus the subroutine operates only on this sparsified structure. The complete description of our actual rebuild subroutine is provided in \cref{sec:rebuild-description}.

Thus, the rebuild subroutine reduces the instance to one in which every relevant edge is incident to a vertex in $S \cup T'$, and the total number of such edges is $O\left(\frac{\alpha}{\epsilon}|V_{\le k}|\right)$.

\begin{algorithm}\caption{A simplified version of our rebuild subroutine}\label{alg:rebuild-simplified}
  Let $T' = T_{\le k}$ and $V' = V_{\le k}$\;
  Clean up dead weight and deviation up to level $k$\;
  \lForEach{$v \in V'$ and $u \in N[v] \setminus I(v)$}{Add $v$ to $F'(u)$}
  Let $S$ be the set of vertices with $|F(v)| + |F'(v)| > (1+\epsilon)\alpha$\;
  \lForEach{$v \in S \cup T'$}{
    $F(v) \gets \emptyset$
  }
  \lForEach{$v \notin S$ such that $F'(v) \ne \emptyset$}{
    $F(v) \gets F(v) \cup F'(v)$
  }
  Set $\dom(v) \gets k$ for all $v \in S \cup T'$ and $\Dom(v) \gets k$ for all $v \in V'$\;
  Call $\water(k, S \cup T', V')$ considering only $N[v] \setminus I(v)$ for each $v \in V_{\le k}$\;
\end{algorithm}

Let $f_v = |F(v)|$ at the beginning of the rebuild subroutine, and let $f_v' = |F'(v)|$ at the end of it. Next, we prove the following claim about the runtime of the rebuild subroutine.

\begin{claim}\label{cl:rebuild-runtime-simplified}
  The runtime of the rebuild subroutine, ignoring the cost of cleaning up deviation up to level $k$, is $O\!\left(k + \frac{\alpha}{\epsilon}|V_{\le k}| + \alpha |T_{\le k}| + \sum_{v \notin S \cup T_{\le k}} f_v'\right)$.
\end{claim}
\begin{proof}
At the moment of the call to $\water$, every vertex in $V'$ is incident only on vertices from $S \cup T'$, so $\sum_{v \in V'} |N[v] \setminus I(v)| \le \sum_{v \in S \cup T'} f_v + f_v'$. Note that we can compute the set $S$ in time $O\!\left(\sum_{v \in V} f_v'\right)$, since we only need to consider vertices with nonempty $F'(v)$.
Observe that the runtime of everything up to the call to $\water$ can be expressed in terms of $f_v$ and $f_v'$ as
$
O\!\left(k + |V_{\le k}| + \sum_{v \notin S \cup T_{\le k}} f_v' + \sum_{v \in S \cup T_{\le k}} f_v + f_v'\right)
$.

According to \cref{lm:high-deg}, the size of $S$ can be bounded by $|S| \le \frac{1+\epsilon}{\epsilon}|V_{\le k}|$. By \cref{inv:forgotten-bound-intro}, $f_v \le (1+\epsilon)\alpha$. Thus, $\sum_{v \in S \cup T_{\le k}} f_v = O(\alpha |S| + \alpha |T_{\le k}|)$.

Since the graph has arboricity at most $\alpha$, we have $\sum_{v \in S} f_v' \le \alpha(|S| + |V_{\le k}|) = O\!\left(\frac{\alpha}{\epsilon}|V_{\le k}|\right)$. Finally, note that for $v \in T_{\le k} \setminus S$, we have $f_v' \le (1+\epsilon)\alpha$, and so $\sum_{v \in T_{\le k} \setminus S} f_v' = O(\alpha |T_{\le k}|)$. Thus, we obtain the claimed runtime bound.
\end{proof}

We note that in the actual algorithm, $F'(v)$ can contain vertices outside $V_{\le k}$. Because of this, we also bound the runtime in terms of passive vertices and outdated copies, which we do not explain here; these terms are handled in the full analysis in \cref{sec:rebuild-analysis}.

We proceed to bound the amortized update time of the rebuild subroutine, and in particular the term $O(\sum_{v \notin S \cup T_{\le k}} f_v')$.

\paragraph*{Amortized analysis of the simplified rebuild}
We analyze the amortized update time using a potential function defined formally in \cref{sec:analysis}. Here we provide an analysis of our simplified rebuild subroutine. The actual analysis follows the same idea, but has to account for more details, since the actual rebuild subroutine is more complicated. In particular, our analysis here ignores the presence of passive vertices.

The goal is to prove the following lemma.
\begin{lemma}\label{lm:intro-potential-decrease}
  The decrease in potential due to rebuild is $\ge \sum_{v \notin S \cup T_{\le k}} f_v' + \frac{\alpha}{\epsilon}|V_{\le k}| + \alpha |T_{\le k}|$.
\end{lemma}
This implies that the amortized runtime of our rebuild subroutine is $O(k)$. This $O(k)$ runtime cost can also be charged to the decrease in potential. This is done using the level potential defined in \cref{sec:analysis}.

First, we lower bound the decrease in potential due to cleaning up deviation up to level $k$. This decrease can be summarized by the following lemma, whose proof is deferred to \cref{sec:rebuild-analysis}. The actual statement there is different and includes some additional terms.
\begin{lemma}\label{thm:down-intro}
  The decrease in potential due to cleaning up dead weight and deviation up to level $k$ is at least $\frac{5(\alpha+1)}{\epsilon}|V_{\le k}| + \frac{5(\alpha+1)}{\epsilon}|T_{\le k}|$.
\end{lemma}
The proof of this lemma is similar to \cite{bhattacharya2021dynamic}. However, as the runtime bound from the previous section, unlike \cite{bhattacharya2021dynamic}, this lemma alone is not enough to bound the amortized runtime cost.

We introduce the following potential function, which intuitively allows us to cover the costs associated with forgetting neighbors of $v$.
For each vertex $v$, we define the {\em forget potential} as $\Phi_{\forget}(v) = -|F(v)|$. The \emph{total forget potential} is defined as $\Phi_{\forget} = \sum_{v \in V} \Phi_{\forget}(v)$.

\begin{proof}[Proof of \Cref{lm:intro-potential-decrease}]
  Consider a vertex $v \notin S \cup T_{\le k}$. By the algorithm, we set $F(v) \gets F(v) \cup F'(v)$, so $\Phi_{\forget}(v)$ decreases by $f'_v$. Consider a vertex $v \in S \cup T_{\le k}$. For $v$, we set $F(v) \gets \emptyset$, and so $\Phi_{\forget}(v)$ increases by $f_v$. According to \cref{inv:forgotten-bound-intro}, $f_v \le (1+\epsilon)\alpha$, which is at most $2\alpha$. Finally, from \cref{lm:high-deg}, we have $|S| \le \frac{1+\epsilon}{\epsilon}|V_{\le k}| \le \frac{2}{\epsilon}|V_{\le k}|$. Thus the increase in the total forget potential is at most
  $
    \frac{4\alpha}{\epsilon}|V_{\le k}| + 2\alpha|T_{\le k}| - \sum_{v \notin S \cup T_{\le k}} f_v'$.
  Summing this with the bound from \cref{thm:down-intro} yields the claimed bound.
\end{proof}

\section{Our Dynamic Algorithm}\label{sec:algo}
Let $\epsilon \in (0, 1/3)$ be a constant. Define $L = \lceil\log_{1+\epsilon} Cn\rceil + 1$. Following the primal-dual framework from \cite{bhattacharya2019new,bhattacharya2021dynamic}, each vertex $v \in V$ is assigned a \emph{dominating level} $\dom(v) \in [L]$. The \emph{dominated level} of a vertex $v \in V$ is defined as $\Dom(v) = \max_{u \in N[v]} \dom(u)$.

Define $\lambda = \frac{1}{(1+\epsilon)^2(3\alpha+1)}$, $\xi = \frac{\alpha+1}{(1+\epsilon)(3\alpha+1)}$ and $\gamma = \frac{2\alpha+1}{3\alpha+1}$. Each vertex $v \in V$ is assigned a level $\lev(v) \in [L]$, and the \emph{packing value} of $v$ is defined as $x_v = \wat{\lev(v)}$. The \emph{true weight} of a vertex $v \in V$ is defined as $\wts(v) = \sum_{u \in N[v]} x_u$. In addition, each vertex $v$ is assigned \emph{dead weight} $\phi(v)$.

Let $\tau_v = \min_{u \in N[v]} c_u$ be the minimum cost in $N[v]$. For each vertex $v$ we define the \emph{base level} $\bs(v) = \left\lfloor\log_{1+\epsilon}\frac{1}{c_v}\right\rfloor$ and the \emph{bottom level} $\bt(v) = \max_{u \in N[v]} \bs(u)$. Throughout the algorithm, we will maintain $\lev(v) \ge \bt(v)$, which implies $x_v \le (1+\epsilon)\lambda \tau_v$ (see \cref{cor:x-ub}). Let $H$ be the set of vertices $v$ such that $\lev(v) = \bt(v)$, which we call the set of \emph{heavy} vertices.

\paragraph{Lazy levels}
Following the framework of \cite{bhattacharya2019new}, we categorize vertices into \emph{active} and \emph{passive}:
\begin{itemize}
\item \textbf{Active.} A vertex $v$ is called \emph{active} if $\lev(v) = \Dom(v)$. Throughout the algorithm, we will make sure that $\dom(v) \ge \bs(v)$ for each vertex $v$ (this is enforced by \cref{inv:tight} defined later). Thus, vertices from $H$ are always active. Let $A_i$ be the set of active vertices with $\lev(v) = i$.
\item \textbf{Passive.} A vertex $v$ is called \emph{passive} if $\Dom(v) < \lev(v)$. Due to the running time issues, we are not able to track $\Dom(v)$ exactly. For such vertices, we keep a value $\zlev(v) \le \Dom(v)$, similar to \cite{bukov2025nearly}, which we call the \emph{lazy level} of $v$. Let $P_i$ be the set of all passive vertices $v$ with $\zlev(v) = i$.
\end{itemize}

\paragraph{Forgotten neighbors}
Throughout the algorithm, we allow vertices to ``forget'' some of their neighbors and stop keeping track of their levels. Let $F(v) \subseteq N[v]$ be the set of \emph{forgotten neighbors} for vertex $v$ (note that a vertex can forget itself).
Additionally, let $I(v)$ be the set of \emph{ignoring neighbors} of $v$. That is, $I(v)$ consists of vertices $u \in N[v]$ such that $v \in F(u)$. Note that we allow a vertex to forget itself.

\paragraph{Viewed levels}
The algorithm does not maintain the true weight $\wts(v)$ of every vertex explicitly. Indeed, when the level $\lev(v)$ of a vertex $v$ increases, the true weight of each neighbor $u \in N[v]$ decreases, and updating these values exactly could require scanning up to $\Delta+1$ vertices.

To avoid this, for every vertex $v$ and every neighbor $u \in N[v] \setminus F(v)$, we maintain a \emph{viewed level} $\view{v}{u} \le \lev(u)$, representing the level at which $v$ currently views $u$. The corresponding \emph{viewed value} is $\viewx{v}{u}=\wat{\view{v}{u}}$, so $\viewx{v}{u}\ge x_u$. We then define the \emph{viewed weight} of $v$ by $\hat{\wts}(v)=\sum_{u\in N[v] \setminus F(v)} \viewx{v}{u}$. Thus, when $\lev(v)$ increases, its neighbors may keep their old viewed level of $v$; that is, they may store an outdated copy of $\lev(v)$.

\paragraph{Outdated copies and deviation}
For each vertex $v$ and level $0 \le i \le L$, we maintain the set $\hat{N}_i[v]$ of vertices $u \in N[v] \setminus \Ign{v}$ such that $\view{u}{v} = i$ as a linked list. Thus, $\hat{N}_i[v]$ is the set of non-ignoring neighbors that view $v$ at level $i$. For each $v$, pointers to the lists $\{\hat{N}_i[v]\}_{0 \le i \le L}$ are stored in an array of length $L+1$.

If $i < \lev(v)$ and $\hat{N}_i[v] \neq \emptyset$, we call $\hat{N}_i[v]$ the \emph{outdated copy of $v$ at level $i$}; the list $\hat{N}_{\lev(v)}[v]$ is called the \emph{primary copy} of $v$. Note that the primary copy may be empty. For each level $i$, let $Q_i$ be the set of vertices $v$ such that $\lev(v) > i$ and $\hat{N}_i[v] \neq \emptyset$, i.e., the set of all vertices that have an outdated copy at level $i$. We store each $Q_i$ as a linked list, with pointers to these lists in an array of length $L+1$. In addition, for every outdated copy of $v$ at level $i$, we store a pointer to the corresponding entry of $v$ in $Q_i$, so that $v$ can be removed from $Q_i$ in constant time.

We define \emph{deviation at level $i$} as $\delta_i = \wat{i}\cdot |Q_i|$ and the \emph{total deviation} as $\delta = \sum_{i=0}^L \delta_i$.

\begin{definition}
  A vertex $v$ is called \emph{tight} when $\hat{\wts}(v) + \phi(v) > \frac{\gamma c_v}{1+\epsilon}$; otherwise, it is called \emph{slack}.
\end{definition}

\begin{assumption}\label{as:dead-weight}
  Throughout the algorithm, we assume that slack vertices have $\phi(v) = 0$. This is easy to enforce, since whenever a vertex becomes slack, we can set $\phi(v)$ to 0. As we will show later, updating the relevant data structures takes constant time.
\end{assumption}

\paragraph{Basic data structures}
Throughout the algorithm, let $T \subseteq V$ be the set of tight vertices. For each level $i$, let $D_i$ be the set of vertices $v$ with $\dom(v)=i$, let $T_i = T \cap D_i$, let $\phi_i = \sum_{v \in D_i} \phi(v)$. We maintain each of the sets $V_i$, $A_i$, $P_i$, $D_i$, $T_i$, and $Q_i$ as a linked list, and we store the collections $\{V_i\}_{0\le i\le L}$, $\{A_i\}_{0\le i\le L}$, $\{P_i\}_{0\le i\le L}$, $\{D_i\}_{0\le i\le L}$, $\{T_i\}_{0\le i\le L}$, $\{Q_i\}_{0\le i\le L}$, together with the values $\{\phi_i\}_{0\le i\le L}$, $\{\delta_i\}_{0\le i\le L}$, $\{c(T_i)\}_{0 \le i \le L}$ and $\{\sum_{v \in V_i} x_v\}_{0 \le i \le L}$, in arrays of length $L+1$. For each vertex $v$, we store pointers to its locations in these lists; for the lists $Q_i$, we store these pointers in an array of length $L+1$. Hence, whenever $\lev(v)$, $\zlev(v)$, $\dom(v)$, $\hat{\wts}(v)$, $\phi(v)$, or some $\delta_i$ changes, all affected lists and values can be updated in constant time.

For each vertex $v$, we store the vertices of $N[v]$ in a balanced binary search tree ordered by cost. This allows us to maintain $\tau_v$, $\bt(v)$, and the sets $H$ in $O(\log n)$ time per edge insertion or deletion. In addition, we maintain a list $H'$ which consists of cheapest neighbors of vertices in $H$ (for each $v \in H$ we take only a single cheapest neighbor $u \in N[v]$). For each $v \in H'$ we count for how many vertices from $H$ it is selected as a cheapest neighbor. This counter allows us to update $H'$ whenever $H$ changes in constant time.

For each vertex $v$, we also maintain the sets $N[v]$, $N[v]\setminus \Ign{v}$, $\Ign{v}$, and $F(v)$ as linked lists. For each edge $\{u, v\}$, we keep pointers to the locations of $u$ in these lists for $v$, and symmetrically for $u$. Thus edge deletions, as well as forgetting or recalling a neighbor, can be handled in constant time.

\subsection{Invariants}\label{sec:invariants}
In this section we describe the invariants that our algorithm aims to satisfy.

\begin{invariant}[Bounded Weight invariant]\label{inv:weight-bound}
  For every vertex $v$, we have $\hat{\wts}(v) \le c_v - |F(v)| \cdot (1+\epsilon)\lambda c_v$.
\end{invariant}

\begin{invariant}[Tightness Invariant]\label{inv:tight}
  For every vertex $v$, we have $\dom(v) \ge \bs(v)$. Moreover, if $\dom(v) > \bs(v)$, then $v$ is tight.
\end{invariant}

\begin{invariant}[Forgotten Neighbors Invariant]\label{inv:forgotten-bound}
  For every vertex $v$, we have $\lvert\Forg{v}\rvert \le (1+\epsilon)\alpha$. Moreover, if $\dom(v) > \bs(v)$, then $\Forg{v} = \emptyset$.
\end{invariant}

\begin{invariant}[Bounded $\delta,\phi$ Invariant]\label{inv:error-bound}
  $(\alpha+1)\delta + \phi \le \epsilon\bigl(\xi c(T) +
  (\alpha+1)\cdot \sum_{v \in V}x_v\bigr).$
\end{invariant}

\begin{invariant}[Level Invariant]\label{inv:level}
  For every vertex $v$, we have $\Dom(v) \le \lev(v)$.
\end{invariant}

\begin{corollary}\label{cor:x-ub}
  For any vertex $v$ we have $\lev(v) \ge \bt(v)$ and $x_v \le (1+\epsilon)\lambda \tau_v$.
\end{corollary}
\begin{proof}
  According to \cref{inv:tight} and the definition of $\Dom(v)$, we have $\Dom(v) \ge \bt(v)$, and so $\lev(v) \ge \bt(v)$ by \cref{inv:level}. To conclude the proof, recall that $x_v = \wat{\lev(v)}$.
\end{proof}

\begin{invariant}[View Invariant]\label{inv:view}
  For every vertex $v$ and every $u \in N[v] \setminus \Ign{v}$, we have $\dom(u) \le \view{u}{v} \le \lev(v)$.
\end{invariant}

\begin{corollary}\label{cor:view-neighbor}
  For any $v \in V$ and $u \in N[v] \setminus F(v)$, we have $\view{v}{u} \ge \dom(v)$.
\end{corollary}
\begin{proof}
  Notice that if $u \notin \Forg{v}$, then $v \notin \Ign{u}$. Therefore, $\dom(v) \le \view{v}{u}$ by \cref{inv:view}.
\end{proof}

\begin{invariant}[Lazy Level Invariant]\label{inv:zlev}
  For every passive vertex $v$, we have $(N[v] \setminus I(v)) \cap T \ne \emptyset$ and $\zlev(v) \le \max_{u \in (N[v]\setminus I(v)) \cap T} \dom(u)$.
\end{invariant}
In other words, \cref{inv:zlev} means that for every passive vertex $v$, there is always a tight neighbor among $N[v] \setminus I(v)$ at dominating level $\zlev(v)$ or above. Thus, passive vertices are always dominated by $T$.

\begin{lemma}
  If \cref{inv:tight,inv:zlev} hold, then $T \cup H'$ forms a dominating set.
\end{lemma}
\begin{proof}
  Consider an active vertex $v$. If $v \in H$, then $v$ is dominated by $H'$. Otherwise, there is a vertex $u \in N[v]$ with $\dom(u) = \Dom(v) = \lev(v)$, and it must be tight by \cref{inv:tight}. Next, consider a passive vertex $v$. For \cref{inv:zlev} to hold, there must be at least one tight vertex in $N[v] \setminus I(v)$, and so $v$ is dominated by $T$.
\end{proof}

\subsection{Approximation guarantee}\label{sec:approximation}
In this section we assume that all invariants from \Cref{sec:invariants} hold. Let $\OPT$ be the cost of an optimal dominating set.

\begin{lemma}\label{lm:feasibility}
  $\{x_v\}_{v \in V}$ is a feasible packing.
\end{lemma}
\begin{proof}
  Consider any vertex $v$. Then,
  \[
    \wts(v) = \sum_{u \in N[v]} x_u = \sum_{u \in N[v] \setminus F(v)} x_u + \sum_{u \in F(v)} x_u.
  \]
  By \cref{cor:view-neighbor,inv:weight-bound} we have
  \[
    \sum_{u \in N[v] \setminus F(v)} x_u \le \sum_{u \in N[v] \setminus F(v)} \viewx{v}{u} \le c_v - |F(v)| \cdot (1+\epsilon)\lambda c_v.
  \]
  By \cref{cor:x-ub} we have
  \[
    \sum_{u \in F(v)} x_u \le |F(v)| \cdot (1+\epsilon)\lambda \tau_v \le |F(v)| \cdot (1+\epsilon)\lambda c_v.
  \]
  Therefore, $\wts(v) \le c_v$.
\end{proof}

\begin{lemma}\label{lm:duality}
  $\sum_{v \in V} x_v \le \OPT$.
\end{lemma}
\begin{proof}
  Let $D^{*}$ be a dominating set of cost $\OPT$. Then, according to \cref{lm:feasibility},
  \[
    \sum_{v \in V} x_v
    \le \sum_{v \in D^{*}} \sum_{u \in N[v]} x_u
    \le \sum_{v \in D^{*}} c_v
    = \OPT.
  \]
\end{proof}

\begin{lemma}\label{lm:H-cost}
  $c(H') \le (1+5\epsilon)(3\alpha+1) \sum_{v \in H} x_v$.
\end{lemma}
\begin{proof}
  By definition, for every $v \in H$ we have $\lev(v) = \bt(v)$. Thus,
  \[
    x_v = \wat{\bt(v)} \ge \lambda \tau_v = \frac{\tau_v}{(1+\epsilon)^2(3\alpha+1)},
  \]
  and hence
  \[
    c(H') = \sum_{v \in H'} c_v \le \sum_{v \in H} \tau_v \le (1+\epsilon)^2(3\alpha+1)\sum_{v \in H} x_v.
  \]
  To conclude the proof, note that $(1+\epsilon)^2 \le 1+5\epsilon$, since $\epsilon \le 1$.
\end{proof}

\begin{definition}[Deviation for a vertex]
  For a vertex $v$, let $Q(v)$ be the set of all levels $i$ such that $v \in Q_i$. Then the \emph{deviation for vertex $v$} is defined as $\delta(v) = \sum_{i \in Q(v)} \wat{i}$.
\end{definition}
Observe that this definition is consistent with the definition of $\delta$, namely $\delta = \sum_{i = 0}^L \delta_i = \sum_{v \in V} \delta(v)$.

\begin{lemma}\label{lm:view-delta}
  For each vertex $v$ and each neighbor $u \in N[v] \setminus F(v)$, we have $\viewx{v}{u} \le x_u + \delta(u)$.
\end{lemma}
\begin{proof}
  Consider a vertex $v$ and some neighbor $u \in N[v] \setminus F(v)$. If $\view{v}{u} = \lev(u)$, then $\viewx{v}{u} = x_u$, and hence the inequality holds. Otherwise, assume that $\view{v}{u} < \lev(u)$. In that case, $\hat{N}_{\view{v}{u}}[u] \ne \emptyset$, since $v \notin I(u)$. Hence, $u \in Q_{\view{v}{u}}$ and $\delta(u) \ge \wat{\view{v}{u}} = \viewx{v}{u}$.
\end{proof}

\begin{observation}\label{obs:orientation}
  The edges of any graph $G$ of arboricity at most $\alpha$ can be oriented so that every vertex has out-degree at most $\alpha$.
\end{observation}
\begin{proof}
  Since $G$ has arboricity at most $\alpha$, its edges can be partitioned into at most $\alpha$ edge-disjoint forests. The edges of a forest can be oriented in such a way that each vertex has at most one outgoing edge. This can be done by fixing a root vertex in each tree of the forest and orienting every edge toward this root. Thus, applying such an orientation to each forest in the partition yields the required orientation.
\end{proof}

\begin{lemma}\label{lm:approximation-guarantee}
  $c(T) + c(H') \le (1 + 5\epsilon)\cdot 2(3\alpha + 1)\cdot\OPT$.
\end{lemma}
\begin{proof}
  Consider the orientation from \Cref{obs:orientation}. For a vertex $v$, let $N^{in}(v)$ and $N^{out}(v)$ denote its in- and out-neighbors, respectively (and let $N^{in}[v]$ and $N^{out}[v]$ include $v$). Then
  \begin{equation}\label{eq:wts-v}
    \hat{\wts}(v)=\sum_{u\in N^{in}[v] \setminus F(v)}\viewx{v}{u}+\sum_{u\in N^{out}(v) \setminus F(v)}\viewx{v}{u}
    \le\sum_{u\in N^{in}[v] \setminus F(v)}(x_u+\delta(u))+(1+\epsilon)\alpha\lambda c_v,
  \end{equation}
  where the last inequality uses \cref{cor:view-neighbor,lm:view-delta}. Summing over $v \in T$ and swapping the order of summation gives
  \begin{equation}\label{eq:in-sum}
    \sum_{v\in T}\sum_{u\in N^{in}[v] \setminus F(v)}(x_u+\delta(u))
    =\sum_{u\in V}(x_u+\delta(u))\cdot|\{v\in T:v\in N^{out}[u]\setminus I(u)\}|
    \le(\alpha+1)\sum_{v\in V}(x_v+\delta(v)).
  \end{equation}
  Recall that for a tight vertex $v\in T$, we have $\hat{\wts}(v)+\phi(v)>\frac{\gamma c_v}{1+\epsilon}$.
  Summing over $T$ and using \eqref{eq:wts-v} and \eqref{eq:in-sum} yields
  \[
    \gamma c(T)\le (1+\epsilon)\sum_{v\in T}\hat{\wts}(v)+\phi(v)
    \le(1+\epsilon)(\alpha+1)\sum_{v\in V}x_v+(1+\epsilon)((\alpha+1)\delta+\phi) + \sum_{v\in T}\frac{\alpha}{3\alpha+1}c_v,
  \]
  hence
  \[
    c(T)\le(1+\epsilon)(3\alpha+1)\sum_{v\in V}x_v+\frac{(1+\epsilon)(3\alpha+1)}{\alpha+1}\big((\alpha+1)\delta+\phi\big).
  \]
  By \Cref{inv:error-bound}, the second term is at most $\epsilon c(T)+\epsilon(1+\epsilon)(3\alpha+1)\sum_{v\in V}x_v$, so
  \[
    (1-\epsilon)c(T)\le(1+\epsilon)^2(3\alpha+1)\sum_{v\in V}x_v.
  \]
  Since $\epsilon\le 1/3$, we have $(1+\epsilon)^2/(1-\epsilon)\le 1+5\epsilon$. Applying \cref{lm:H-cost,lm:duality} yields
  \[
    c(T)+c(H')\le(1+5\epsilon)\cdot 2(3\alpha+1)\sum_{v\in V}x_v\le(1+5\epsilon)\cdot 2(3\alpha+1)\cdot\OPT,
  \]
  as required.
\end{proof}

\subsection{Algorithm description}\label{sec:alg-description}
We describe the subroutines $\ins(e)$, $\del(e)$, and $\reset(k)$, which constitute the main algorithm. The pseudocode of the main algorithm is given in \Cref{alg:main}. At the beginning of the algorithm, we assume that the graph is empty. Each vertex $v$ is assigned $\view{v}{v} = \lev(v) = \dom(v) = \bs(v)$ and is active. Note that initially we have $H' = H = V$, and thus every vertex is dominated by itself. When an edge $e$ is inserted into the graph, we call the subroutine $\ins(e)$ to handle it; when $e$ is deleted from the graph, we call the subroutine $\del(e)$. After that, we check whether \Cref{inv:error-bound} is violated as a result. If it is, then we search for the smallest level $k$ such that $(\alpha+1)\delta_{\le k} + \phi_{\le k} > \epsilon(\xi c(T_{\le k}) + (\alpha+1)\sum_{v \in V_{\le k}} x_v)$ and call the subroutine $\reset(k)$. We repeat this in a while loop until the invariant is repaired. Note that we can find such a level $k$ in time $O(k)$, since we maintain the quantities $\delta_i$, $\phi_i$, $c(T_i)$, and $\sum_{v \in V_i} x_v$ for each level $i$. We will argue later in the amortized runtime analysis that this $O(k)$ cost can be charged to the decrease in potential due to $\reset(k)$.

\begin{algorithm}\caption{$\mathsf{DynamicDominatingSet}$}\label{alg:main}
  Initialize $\view{v}{v}, \lev(v), \dom(v) \gets \bs(v)$ for each $v \in V$\;
  \ForEach{edge update $e$}{
    \If{$e$ is inserted}{
      $\ins(e)$\;
    }\Else{
      $\del(e)$\;
    }
    \While{~\Cref{inv:error-bound} is violated}{
      Find the smallest level $k$ such that $(\alpha+1)\delta_{\le k} + \phi_{\le k} > \epsilon(\xi c(T_{\le k}) + (\alpha+1)\sum_{v \in V_{\le k}} x_v)$\;
      $\reset(k)$\;
    }
  }
\end{algorithm}

\subsubsection{Insert}\label{sec:insert}
Refer to \Cref{alg:insert} for the pseudocode of the subroutine $\ins(e)$ and to \cref{fig:insert} for an illustration.

Suppose an edge $e = \{u,v\}$ gets inserted. Consider, without loss of generality, the endpoint $v$.
The insertion could lead to an increase in $\Dom(v)$. Because of that, we might also need to increase $\lev(v)$ to maintain the level invariant (\Cref{inv:level}). Since we introduce a new vertex into the neighborhood, we also need to satisfy $\dom(u) \le \view{u}{v} \le \lev(v)$ (\Cref{inv:view}), and $\view{u}{v}$ should be large enough to satisfy $\hat{\wts}(u) \le c_u - |F(u)| \cdot (1+\epsilon) \lambda c_u$ (\Cref{inv:weight-bound}). A way to do that is to pick the lowest possible $h_v \ge \dom(u)$ such that $\hat{\wts}(u) + \wat{h_v} \le c_u - |F(u)| \cdot (1+\epsilon) \lambda c_u$ and $h_v \ge \lev(v)$, and set $\lev(v)$ and $\view{u}{v}$ to $h_v$; we can find such an $h_v$ in time $O(L) = O(\log_{1+\epsilon} (Cn))$. The first condition guarantees that \Cref{inv:weight-bound} is satisfied. The second condition makes $v$ satisfy \Cref{inv:level}. Since we do not keep track of $\Dom(v)$, we rely on the fact that $\lev(v) \ge \Dom(v)$ before the insertion (this is due to \cref{inv:level}), which together with $h_v \ge \dom(u)$ gives us $\lev(v) \ge \Dom(v)$ after the insertion. The view invariant (\cref{inv:view}) is satisfied, since $\lev(v)$ might have only increased, and for the freshly inserted edge we have $\view{u}{v} = \lev(v)$ and $\lev(v) \ge \dom(u)$.

By choosing the smallest possible value, we ensure that the lazy level invariant (\cref{inv:zlev}) is satisfied. Indeed, suppose $v$ is passive after the insertion. Observe that the algorithm does not change the weight of any vertex in $N[v] \setminus \{u\}$. Thus, if $v$ was passive before, then the invariant still holds. Otherwise, if $v$ was active and has become passive now, then $h_v > \max\{\dom(u),\lev(v)\}$. But in that case, $\hat{\wts}(u) + \wat{(h_v - 1)} > c_u - |F(u)| \cdot (1+\epsilon) \lambda c_u \ge \gamma c_u$, since $|F(u)| \le (1+\epsilon)\alpha$ by \cref{inv:forgotten-bound}, and so $\hat{\wts}(u) + \wat{h_v} > \gamma c_u/(1+\epsilon)$ and hence $u$ is tight.

Observe that the forgotten neighbors invariant (\cref{inv:forgotten-bound}) is maintained, since $\dom(v)$ and $F(v)$ remain the same for all $v \in V$.

\begin{remark}[Maintaining data structures]
  Since we inserted a new edge, we need to insert $v$ into $\hat{N}_{h_v}$. Note that this might require us to insert $u$ into $Q_{h_v}$ if $\hat{N}_{h_v}$ was empty. All of this can be done in constant time.
\end{remark}

\begin{algorithm}\caption{$\ins(e)$}\label{alg:insert}
  \tcp{Let $e=\{u,v\}$}
  Find the lowest level $h_v \ge \max\{\dom(u), \lev(v)\}$ such that $\hat{\wts}(u) + \wat{h_v} \le c_u - |F(u)| \cdot (1+\epsilon) \lambda c_u$\;
  $\lev(v), \view{u}{v} \gets h_v$\;
  $\hat{\wts}(u) \gets \hat{\wts}(u) + \wat{h_v}$\;
  $\zlev(v) \gets 0$\;
  Repeat the same for $u$ (with the roles of $v$ and $u$ flipped)\;
\end{algorithm}

\begin{figure}[t]
  \centering
  \begin{tikzpicture}[scale=1, every node/.style={font=\small}]
\begin{scope}[shift={(0,0)}]
    \foreach \y/\lab in {0/0,1/1,2/2,3/3,4/4} {
      \draw[dashed,gray] (0,\y) -- (4,\y);
      \node[left] at (0,\y) {\lab};
    }

    \node[circle,fill=black,inner sep=2.8pt,label=above:{$\hat{N}_4[v]$}] (v4) at (2.4,4) {};
    \node[circle,draw,inner sep=2.8pt,label=above:{$\hat{N}_2[v]$}] (v2) at (2.4,2) {};
    \node[circle,draw,inner sep=2.8pt,label=below:$v$,label=above:{$\hat{N}_0[v]$}] (v0) at (2.4,0) {};

    \node at (1,5.0) {$\lev(v)=4$};

    \node[circle,fill=black,inner sep=1.5pt,label=above:$u$] (u) at (5.8,5.1) {};
    \draw[red,thick,dashed] (u) -- (v4);

    \node[circle,fill=black,inner sep=1.5pt,label=right:$v$] (r1) at (5.8,4.2) {};
    \node[circle,fill=black,inner sep=1.5pt,label=right:$a$] (r2) at (5.8,3.1) {};
    \node[circle,fill=black,inner sep=1.5pt,label=right:$b$] (r3) at (5.8,1.8) {};
    \node[circle,fill=black,inner sep=1.5pt,label=right:$c$] (r4) at (5.8,0.5) {};

    \draw[thick]
      (5.75,2.4) ellipse [x radius=0.95, y radius=2.5];

    \node[below=4pt] at (5.74,-0.3) {$N[v]\setminus I(v)$};

    \draw[thick] (v4) -- (r1);
    \draw[thick] (v2) -- (r2);
    \draw[thick] (v2) -- (r3);
    \draw[thick] (v0) -- (r4);
  \end{scope}

  \begin{scope}[shift={(8.2,0)}]
    \foreach \y/\lab in {0/0,1/1,2/2,3/3,4/4,5/5} {
      \draw[dashed,gray] (0,\y) -- (4,\y);
      \node[left] at (0,\y) {\lab};
    }

    \node[circle,fill=black,inner sep=2.8pt,label=above:{$\hat{N}_5[v]$}] (v5b) at (2.4,5) {};
    \node[circle,draw,inner sep=2.8pt,label=above:{$\hat{N}_4[v]$}] (v4b) at (2.4,4) {};
    \node[circle,draw,inner sep=2.8pt,label=above:{$\hat{N}_2[v]$}] (v2b) at (2.4,2) {};
    \node[circle,draw,inner sep=2.8pt,label=below:$v$,label=above:{$\hat{N}_0[v]$}] (v0b) at (2.4,0) {};

    \node at (1,6.0) {$\lev(v)=5$};

    \node[circle,fill=black,inner sep=1.5pt,label=above:$u$] (ub) at (5.8,5.1) {};
    \draw[red,thick] (ub) -- (v5b);

    \node[circle,fill=black,inner sep=1.5pt,label=right:$v$] (r1b) at (5.8,4.2) {};
    \node[circle,fill=black,inner sep=1.5pt,label=right:$a$] (r2b) at (5.8,3.1) {};
    \node[circle,fill=black,inner sep=1.5pt,label=right:$b$] (r3b) at (5.8,1.8) {};
    \node[circle,fill=black,inner sep=1.5pt,label=right:$c$] (r4b) at (5.8,0.5) {};

    \draw[thick]
      (5.75,2.8) ellipse [x radius=0.95, y radius=3];

    \node[below=4pt] at (5.74,-0.3) {$N[v]\setminus I(v)$};

    \draw[thick] (v4b) -- (r1b);
    \draw[thick] (v2b) -- (r2b);
    \draw[thick] (v2b) -- (r3b);
    \draw[thick] (v0b) -- (r4b);
  \end{scope}
\end{tikzpicture}

\caption{An example of the insertion of an edge $e = \{u,v\}$. In this figure we consider the situation only for vertex $v$. The left part depicts the situation before the insertion of $e$. The vertex $v$ has $\lev(v) = 4$ and $N[v] \setminus I(v) = \{v, a, b, c\}$. There are two outdated copies of $v$ at levels 0 and 2, depicted by white circles, and there is a primary copy of $v$ at level 4. The vertex $v$ views itself at level $4$, and thus the primary copy is $\hat{N}_4[v] = \{v\}$. Vertices $a$ and $b$ view $v$ at level $2$, and the outdated copy at level 2 is $\hat{N}_2[v] = \{a,b\}$. Vertex $c$ views $v$ at level $0$, and so the outdated copy at level 0 is $\hat{N}_0[v] = \{c\}$. The vertex $u$ has $\dom(u) = 5$. The situation after the insertion of $e$ is depicted in the right part. After that, $\lev(v)$ becomes 5, and $u$ views $v$ at level 5; the primary copy is now $\hat{N}_5[v] = \{u\}$. Vertex $v$ keeps viewing itself at level $4$, but now $\hat{N}_4[v] = \{v\}$ is an outdated copy. The remaining outdated copies are the same as before the insertion.}
  \label{fig:insert}
\end{figure}
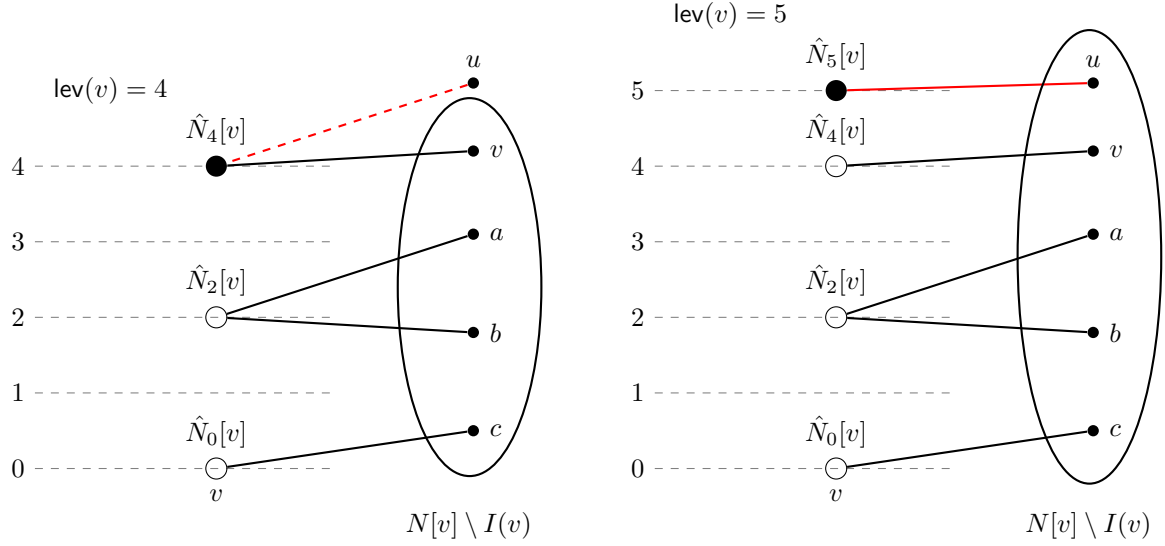

\subsubsection{Deletion}\label{sec:delete}
Refer to \cref{alg:delete} for the pseudocode of the subroutine $\del(e)$.

Suppose an edge $e = \{u,v\}$ gets deleted. Consider, without loss of generality, the endpoint $v$. As a result, $\Dom(v)$ might have decreased and $v$ might have become passive. For the lazy level invariant (\cref{inv:zlev}) to hold, there must be at least one tight neighbor among $N[v] \setminus I(v)$. We cannot afford to scan through the neighbors to verify this, as there might be too many of them. Instead, we make sure that $v$ is dominated by itself. If $\dom(v) > \bs(v)$, then, according to the tightness invariant (\cref{inv:tight}), the vertex $v$ is tight before the deletion. To make sure it is tight afterward, we compensate for the loss in $\hat{\wts}(v)$ by increasing its dead weight $\phi(v)$ by $\wat{\view{v}{u}}$; this maintains \cref{inv:tight}. Otherwise, if $\dom(v) = \bs(v)$, then we make $v$ tight by setting $\phi(v) \gets \gamma c_v$. We note that we cannot simply set $\phi(v) \gets \gamma c_v$ in the case $\dom(v) > \bs(v)$, as this can be too costly in terms of the potential increase (which will become clear after we define the potential functions in \cref{sec:analysis}).

As for the other invariants, observe that $\Dom(v)$, $\hat{\wts}(v)$ and $|F(v)|$ can only decrease, so \cref{inv:level,inv:weight-bound,inv:forgotten-bound} are maintained. As for the view invariant (\cref{inv:view}), it is satisfied trivially, as no relevant variable is changed.

\begin{remark}[Maintaining data structures]
  If $u \in \Ign{v}$, we also need to remove $u$ from $\Ign{v}$. Otherwise, if $u \notin I(v)$, then $u \in \hat{N}_{\view{u}{v}}[v]$, and we need to remove it from that list and possibly remove $u$ from $Q_{\view{u}{v}}$, all of which can be done in constant time.
\end{remark}

\begin{algorithm}\caption{$\del(e)$}\label{alg:delete}
  \tcp{Let $e=\{u,v\}$}
  \If{$\dom(v) > \bs(v)$}{
    $\phi(v) \gets \phi(v) + \viewx{v}{u}$\;
  }
  \Else{
    $\phi(v) \gets \gamma c_v$\;
  }
  \If{$u \notin F(v)$}{
    $\hat{\wts}(v) \gets \hat{\wts}(v) - \viewx{v}{u}$\;
  }\Else{
    $\Forg{v} \gets \Forg{v} \setminus \{u\}$\;
  }
  $\zlev(v) \gets 0$\;
  Repeat the same for $u$ (with the roles of $u$ and $v$ flipped)\;
\end{algorithm}

\subsection{Rebuilding}\label{sec:rebuild}
The $\reset$ subroutine is triggered when \Cref{inv:error-bound} is violated. We assume that $k$ is the smallest level such that
\[
(\alpha+1)\delta_{\le k} + \phi_{\le k} > \epsilon \left( \xi c(T_{\le k}) + (\alpha+1) \sum_{v \in V_{\le k}} x_v \right).
\]

\paragraph{High-level idea}
The goal of the $\reset(k)$ subroutine is to eliminate the dead weights of vertices up to dominating level $k$, eliminate all outdated copies up to level $k$, and then repair any violated invariants. Notice that, by eliminating dead weights and outdated copies, we might decrease $\hat{\wts}(v) + \phi(v)$ for vertices $v \in D_{\le k}$, and thus some of them might become slack and violate the tightness invariant (\cref{inv:tight}). The invariants are then restored using the $\water$ subroutine.

On a high level, before calling the $\water$ subroutine, we raise all vertices from $D_{\le k}$ to dominating level $k+1$ and, as a result, we raise vertices from $V_{\le k}$ to dominated level $k+1$. Note that doing so can only decrease the weights of vertices in $D_{\le k}$. We then process vertices $v \in P_{\le k} \setminus V_{\le k}$ some special way that either makes them active, or makes $\zlev(v) = k+1$. This processing of passive vertices can make some vertices from $D_{\le k}$ tight. We leave such vertices at dominating level $k+1$, and lower the remaining to dominating level $k$. After that, we call the $\water$ subroutine on them. Alas, this approach in not sufficient, since it incurs a $\Delta$ dependency in the runtime.

To improve the runtime of our rebuild subroutine, we try to reduce the number of vertices from $D_{\le k}$ that get raised to dominating level $k$ and participate in the subsequent call to $\water$. As we will show in \cref{sec:analysis}, we can afford to include all vertices in $T_{\le k}$. However, we cannot claim the same for slack vertices, and so we try to filter out some of them and leave only those that we can afford in terms of the amortized runtime (for the details we refer to \cref{sec:analysis}).

This filtering is done using the forgotten neighbors mechanism. Namely, for vertices in $D_{\le k}$, we try to forget vertices from $V_{\le k}$. We first compute, for each vertex $v$, a list $F'(v) \subseteq N[v] \setminus I(v)$ of ``need-to-be-forgotten'' neighbors. We go over vertices $v \in V_{\le k}$ and add $v$ to $F'(u)$ for each $u \in N[v] \setminus I(v)$. After that, we include in the subsequent call to $\water$ only vertices that have many neighbors in $F(v) \cup F'(v)$ (besides the vertices from $T_{\le k}$); specifically, those satisfying $|F(v)| + |F'(v)| > (1+\epsilon)\alpha$. We denote the set of such vertices by $S$. For vertices $v \in S$ (together with vertices $v \in T_{\le k}$), we update the corresponding viewed level $\view{u}{v}$ to the actual $\lev(v)$ for $u \in F(v) \cup F'(v)$ and reset the set $F(v)$ to $\emptyset$. For the remaining vertices (those not in $S \cup T_{\le k}$), we add the neighbors from $F'(v)$ to $F(v)$, thus forgetting them. After this, for every vertex $v \in V_{\le k}$, the set $N[v] \setminus I(v)$ contains only vertices from $S \cup T_{\le k}$. Thus, the subsequent call to $\water$ only needs to consider vertices from $S \cup T_{\le k}$.

We now explain how we clean up outdated copies and how we handle passive vertices. To clean up outdated copies up to level $k$, we employ the same forgetting mechanism. For every vertex $v$ that has an outdated copy at level $k$ or below, we add $v$ to $F'(u)$ for each $u \in N[v] \setminus I(v)$. Recall that we either add all of $F'(v)$ to $F(v)$, or reset $v$, updating all viewed levels of neighbors from $F(v) \cup F'(v)$ to the actual level. Therefore, all outdated copies up to level $k$ are removed. We will later show in \cref{sec:analysis} that we can afford to process all such vertices as well.

We also explain how we handle passive vertices. The issue with passive vertices is that for each such vertex we need to ensure that there is a tight neighbor at level $\zlev(v)$ or above among $N[v] \setminus I(v)$. Since we make some vertices slack as a result of cleaning up outdated copies and dead weight, the lazy level invariant (\cref{inv:zlev}) might be violated for them. Therefore, for each passive vertex $v$ with $\zlev(v) \le k$, we do the following. First, we try to find a tight neighbor $u \in N[v] \setminus I(v)$ with $\dom(u) \ge k + 1$, and if there is one, we set $\zlev(v) \gets k+1$. We do this by iterating over these neighbors, and on each failure, we add $v$ to $F'(u)$. This will later allow us to bound the amortized runtime cost. Intuitively, we gain one credit for adding $v$ to $F'(u)$, so these failed iterations are ``free'' in terms of amortized runtime. We note that it is necessary to stop at the first tight neighbor $u \in N[v] \setminus I(v)$ with $\dom(u) \ge k+1$, since there can be many such neighbors and we are not able to amortize the runtime cost for all of them.

Finally, we explain how we handle passive vertices for which there is no tight neighbor $u \in N[v] \setminus I(v)$ with $\dom(u) \ge k+1$. For such vertices, we try to make them active and include them in the subsequent call to $\water$. This is done by decreasing $\lev(v)$ one level at a time until we reach $k+1$ or $\bt(v)$, or until the following happens. It may be that, by decreasing $\lev(v)$, some neighbor $u \in N[v] \setminus I(v)$ becomes tight. In that case, we set $\zlev(v) \gets k+1$. We argue later that this tight neighbor will eventually satisfy $\dom(u) \ge k+1$, so this is safe to do.

\begin{figure}[t]
  \centering
  \begin{tikzpicture}[scale=1, every node/.style={font=\small}]
\tikzset{
    vtx/.style={circle, fill=black, inner sep=2pt},
    vtc/.style={circle, draw=black, fill=white, inner sep=2pt},
    rededge/.style={red, line width=1.0pt},
    blueedge/.style={blue!80!black, line width=1.0pt},
    blackedge/.style={black, line width=0.9pt},
    skipdots/.style={black, densely dotted, line width=0.9pt},
    connection/.style={gray, dashed}
  }

  \node[vtx] (u) at (0,0.5) {};
  \node[above=6pt of u] {$u$};

  \node[vtx] (u1) at (-3.0,2.5) {};
  \node[vtx] (ui) at (-1.6,2.5) {};
  \node[vtx] (v1) at (-3.0,-2.0) {};
  \node[vtx] (vj) at (-1.6,-2.0) {};
  \node at (-2.3,2.5) {$\cdots$};
  \node at (-2.3,-2.0) {$\cdots$};

  \node[above=6pt of u1] {$u_1$};
  \node[above=6pt of ui] {$u_i$};
  \node[below=6pt of v1] {$v_1$};
  \node[below=6pt of vj] {$v_j$};

  \draw[rounded corners=8pt, black, line width=0.9pt]
    (-3.55,3.5) rectangle (-1.05,-3);
  \node[below=8pt] at (-2.30,0.8) {$F(v)$};

  \draw[rounded corners=8pt, black, line width=0.9pt]
    (1.0,0.6) rectangle (6.5,-3);
  \node[below=8pt] at (3.2,0.5) {$F'(v)$};

  \node[vtx] (u1p) at (1.6,2.5) {};
  \node[vtx] (uip) at (3.0,2.5) {};
  \node[vtx] (v1p) at (1.6,-2.0) {};
  \node[vtx] (vjp) at (3.0,-2.0) {};
  \node[vtc] (w1) at (4.5,-2.0) {};
  \node[vtc] (wl) at (6.0,-2.0) {};
  \node[vtx] (pw1) at (4.5,0.0) {};
  \node[vtx] (pwl) at (6.0,0.0) {};
  \node at (2.3,2.5) {$\cdots$};
  \node at (2.3,-2.0) {$\cdots$};
  \node at (5.25,-2.0) {$\cdots$};
  \node at (5.25,0.0) {$\cdots$};

  \node[above=6pt of u1p] {$u'_1$};
  \node[above=6pt of uip] {$u'_{i'}$};
  \node[below=6pt of v1p] {$v'_1$};
  \node[below=6pt of vjp] {$v'_{j'}$};
  \node[below=6pt of w1] {$w_1$};
  \node[below=6pt of wl] {$w_{l}$};

  \draw[skipdots] (-4.0,-1) -- (6.8,-1);
  \node[right] at (6.85,-1) {$k$};

  \draw[rededge] (u) -- (u1);
  \draw[rededge] (u) -- (ui);
  \draw[rededge] (u) -- (v1);
  \draw[rededge] (u) -- (vj);

  \draw[blackedge] (u) -- (u1p);
  \draw[blackedge] (u) -- (uip);
  \draw[blueedge] (u) -- (v1p);
  \draw[blueedge] (u) -- (vjp);
  \draw[blueedge] (u) -- (w1);
  \draw[blueedge] (u) -- (wl);

  \draw[connection] (w1) -- (pw1);
  \draw[connection] (wl) -- (pwl);

\end{tikzpicture}

  \caption{An example of a vertex $u \in D_{\le k} \setminus V_{\le k}$ during $\reset(k)$. The set of forgotten neighbors $F(u)$ consists of the vertices $v_1,\ldots,v_j$ and $u_1,\ldots,u_i$. The vertices $v_1,\ldots,v_j$ are in $V_{\le k}$, while the vertices $u_1,\ldots,u_i$ are outside $V_{\le k}$. The remaining neighbors of $u$ from $V_{\le k}$ are $v'_1,\ldots,v'_{j'}$, and all of them are in $F'(u)$. The other neighbors of $u$ in $F'(u)$ are $w_1,\ldots,w_l$, and these neighbors either have an outdated copy at level $k$ or below, or their lazy level is at most $k$ (which is depicted by white circles). The remaining neighbors of $u$ are $u'_1,\ldots,u'_{i'}$, and they lie outside both $F(u)$ and $F'(u)$; all of them are outside $V_{\le k}$.}
  \label{fig:rebuild2}
\end{figure}
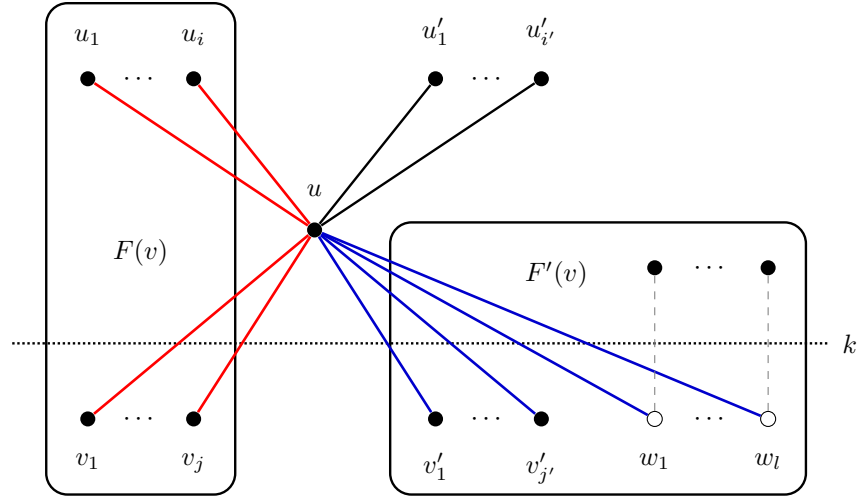

\subsubsection{Description of the $\reset(k)$ subroutine}\label{sec:rebuild-description}
Let us denote by the superscript ``old'' the values of the variables at the beginning of $\reset(k)$. Next, we state a couple of key properties of the $\reset(k)$ subroutine.

\begin{property}\label{prop:dom-above}
  Throughout the $\reset(k)$ subroutine, vertices in $D^{\old}_{\ge k+1}$ stay at the same dominating level. Moreover, vertices from $T^{\old}_{\ge k+1}$ remain tight.
\end{property}

\begin{property}\label{prop:zlev-above}
  Throughout the $\reset(k)$ subroutine, $\zlev(v)$ stays the same for vertices $v \in P^{\old}_{\ge k+1}$.
\end{property}

\begin{property}
  After the $\reset(k)$ subroutine, vertices $v \in P^{\old}_{\le k}$ either become active, or $\zlev(v) \ge k+1$.
\end{property}

\begin{property}
  After the $\reset(k)$ subroutine, vertices $v \in A^{\old}_{\le k}$ are active and $\lev(v) \le k+1$.
\end{property}

\begin{property}\label{prop:cleanup}
  The $\reset(k)$ subroutine eliminates all dead weights for vertices in $T^{\old}_{\le k}$ and all outdated copies at level $k$ or below. Moreover, it does not create new dead weights and outdated copies.
\end{property}

\begin{note}
  \Cref{prop:dom-above} implies that the conditions of \cref{inv:tight} hold for vertices in $D^{\old}_{\ge k+1}$ after the call to $\reset(k)$. Together with \cref{prop:zlev-above}, it implies that the conditions of \cref{inv:zlev} hold for vertices in $P^{\old}_{\ge k+1}$. Thus, intuitively, we only need to focus on everything below level $k+1$.
\end{note}

At the beginning of $\reset(k)$, we classify active vertices with $\lev(v) \le k$ and passive vertices with $\zlev(v) \le k$ as \emph{clean} or \emph{dirty}.
\begin{definition}\label{def:dirty/clean}
  We call a vertex $v$ \emph{clean} if it is active with $\lev(v) \le k$, or if it is passive with $\zlev(v) \le k$ and $\lev(v) \le k+1$; we call a vertex $v$ \emph{dirty} if it is passive with $\zlev(v) \le k$ and $\lev(v) > k+1$.
\end{definition}

Throughout the $\reset$ subroutine, we use a set $F'(v)$ (stored as a linked list) for each vertex $v$, which we assume is initially empty. Instead of initializing this list in every call to $\reset$, we initialize it once at the beginning of the algorithm and clean it up after each $\reset$.

\begin{remark}[Maintaining data structures]
  Throughout the $\reset(k)$ subroutine, we update some $\view{v}{u}$ and $F(v)$. During each such update, we update $\hat{\wts}(v)$, $\hat{N}_i[u]$ and $Q_i$ (for the relevant $i$) accordingly. Notice that this can be done in constant time.
\end{remark}

For convenience, the $\reset(k)$ subroutine is split into several steps. For the pseudocode of the $\reset(k)$ subroutine, refer to \cref{alg:rebuild}.

\begin{step}\label{step:init}
Compute the list of all clean vertices $\mathcal{C}$ and the list of all dirty vertices $\mathcal{D}$.
Initialize sets $\mathcal{D}' \gets \emptyset$ and $T' \gets T_{\le k}$.
\end{step}

\begin{algorithm}\caption{$\reset(k)$}\label{alg:rebuild}
  Compute the lists of clean and dirty vertices $\mathcal{C}$ and $\mathcal{D}$\;
  Initialize $\mathcal{D}' \gets \emptyset$ and $T' \gets T_{\le k}$\;
  Set $\phi(v) \gets 0$ for each $v \in T'$\;
  \ForEach{$0 \le i \le k$}{
    \ForEach{$v \in Q_i$}{
      \ForEach{$u \in \hat{N}_i[v]$}{
        $F'(u) \gets F'(u) \cup \{v\}$\;
      }
    }
  }
  \ForEach{$v \in \mathcal{C}$}{
    \ForEach{$u \in \hat{N}_{\lev(v)}$}{
      $F'(u) \gets F'(u) \cup \{v\}$\;
    }
    Set $\lev(v) \gets k+1$ and make $v$ active\;
  }
  \ForEach{$v \in \mathcal{D}$}{
    \ForEach{$u \in N[v] \setminus \Ign{v}$}{
      \If{$\dom(u) \ge k + 1$ and $u$ is tight}{
        $\zlev(v) \gets k+1$ and continue to the next dirty vertex\;
      }\Else{
        $F'(u) \gets F'(u) \cup \{v\}$\;
      }
    }
    Add $v$ to $\mathcal{D}'$\;
  }
  Let $S$ be the set of vertices with $|F(v)| + |F'(v)| > (1+\epsilon)\alpha$\;
  \ForEach{$v \in S \cup T'$}{
    \ForEach{$u \in F(v) \cup F'(v)$}{
      Update $\view{v}{u}$ to $\lev(u)$\;
    }
    $\Forg{v} \gets \emptyset$\;
    $\dom(v) \gets \max\{\bs(v), k+1\}$\;
  }
  \ForEach{$v \notin S \cup T'$ such that $F'(v) \ne \emptyset$}{
    $\Forg{v} \gets \Forg{v} \cup F'(v)$\;
  }
  \ForEach{$v \in \mathcal{D}'$}{
    Find the minimum $h \ge \max\{k+1, \bt(v)\}$ such that $\hat{\wts}(u) - \viewx{u}{v} + \wat{h} \le c_u - |F(u)| \cdot (1+\epsilon)\lambda c_u$ for all $u \in N[v] \setminus \Ign{v}$\;
    Set $\lev(v)$ to $h$ and update $\view{u}{v}$ to $h$ for each $u \in N[v] \setminus \Ign{v}$\;
    \If{$h = k+1$ or $h = \bt(v)$}{
      Make $v$ active\;
    }\Else{
      $\zlev(v) \gets k+1$\;
    }
  }
  Let $D'$ be the set of slack vertices $v \in S \cup T'$ such that $\dom(v) > \bs(v)$\;
  Let $V'$ be the set of vertices $v \in \mathcal{C} \cup \mathcal{D}'$ such that $\lev(v) > \bt(v)$ and there is no tight neighbor among $N[v] \setminus I(v)$\;
  Set $\dom(v) \gets k$ for every $v \in D'$\;
  \ForEach{$v \in V'$}{
    Set $\lev(v) \gets k$ and update $\view{u}{v}$ to $k$ for every $u \in N[v] \setminus I(v)$\;
  }
  Invoke $\water(k, D', V')$ considering only $N[v] \setminus I(v)$ for $v \in V'$\;
\end{algorithm}

\begin{step}\label{step:remove-dead-weight}
Set $\phi(v) \gets 0$ for each vertex $v \in T'$.
Note that after that we have $\phi_{\le k} = 0$, since recall that we maintain $\phi(v) = 0$ for slack vertices (\cref{as:dead-weight}).
\end{step}

\begin{step}\label{step:handle-outdated}
For each level $0 \le i \le k$ and each vertex $v \in Q_i$, go over the vertices $u \in \hat{N}_i[v]$ and add $v$ to $F'(u)$.
\end{step}

\begin{step}[Handling clean vertices]\label{step:handle-clean}
For each $v \in \mathcal{C}$, add $v$ to $F'(u)$ for each $u \in
\hat{N}_{\lev(v)}$ and set $\lev(v) \gets k+1$ and make $v$ active.
\end{step}

\begin{remark}
Observe that the level invariant and the view invariant (\Cref{inv:level,inv:view}) still hold, since $\lev(v)$ might only have increased.
After \cref{step:handle-clean}, clean vertices are temporarily in an incorrect state, since they are active, but for them $\lev(v) = k+1$ and $\Dom(v) \le k$. This will be fixed in \cref{step:water}, after we call the subroutine $\water$ on them.
\end{remark}

\begin{step}[Filtering dirty vertices]\label{step:filter-dirty}
For each vertex $v \in \mathcal{D}$ scan through $u \in N[v] \setminus I(v)$. Recall that we can efficiently enumerate vertices from $N[v] \setminus \Ign{v}$. If $u$ is tight and $\dom(u) \ge k+1$ then set $\zlev(v) \gets k+1$ and continue to the next vertex in $\mathcal{D}$. Otherwise, add $v$ to $F'(u)$. After scanning through all $u \in N[v] \setminus I(v)$, add $v$ to $\mathcal{D}'$. Thus, $\mathcal{D}'$ consists of dirty vertices for which there is no tight neighbor $u$ with $\dom(u) \ge k+1$ among $N[v] \setminus I(v)$.
\end{step}

\begin{remark}
Consider a vertex $v \in \mathcal{D}$. It is possible that there is already a tight neighbor $u \in N[v] \setminus I(v)$ with $\dom(u) \ge k+1$; in that case we can simply increase $\zlev(v)$ to $k+1$, as tight vertices above level $k$ will stay tight, according to \cref{prop:dom-above}, and so the lazy level invariant (\cref{inv:zlev}) will be satisfied for $v$.
To check if there is such a neighbor, we iterate over neighbors $u \in N[v] \setminus \Ign{v}$ until we find one that is tight and satisfies $\dom(u) \ge k+1$. Every time we iterate over a vertex $u$ that does not satisfy this condition, we add $v$ to $F'(u)$ to compensate for the cost of this iteration. We will argue later in \cref{sec:analysis} that we can cover a constant runtime cost for each vertex in $F'(u)$.
If we fail to find a tight neighbor $u \in N[v] \setminus I(v)$ with $\dom(u) \ge k+1$, then we add $v$ to $\mathcal{D}'$ to deal with it later.
\end{remark}

\begin{step}[Sparsification]\label{step:sparsification}
Let $S$ be the set of vertices with $|F(v)| + |F'(v)| >
(1+\epsilon)\alpha$. Note that to compute $S$ we only need to scan
vertices $v$ such that $F'(v) \ne \emptyset$. For each vertex $v \in S
\cup T'$ update $\view{v}{u}$ to $\lev(v)$ for all $u \in F(v) \cup F'(v)$, set $F(v)\gets \emptyset$ and set $\dom(v) \gets \max\{\bs(v), k+1\}$. For each vertex $v \notin S \cup T'$ such that $F'(v) \ne \emptyset$ add neighbors from $F'(v)$ to $F(v)$.
\end{step}

\begin{remark}
  Consider any vertex $v$ that is processed during \cref{step:sparsification}.
  \begin{enumerate}
  \item If $|F(v)|+ |F'(v)| > (1+\epsilon)\alpha$ or $v \in T'$, then we reset vertex $v$ by emptying its list of forgotten vertices $F(v)$. Note that this preserves the bounded weight invariant (\cref{inv:weight-bound}).
    Additionally, we update the viewed levels of vertices in $F'(v) \cup F(v)$ to their actual levels.
  \item Otherwise, if $|F(v)| + |F'(v)| \le (1+\epsilon)\alpha$ and $v \notin T'$, we forget the neighbors from $F'(v)$. Observe that $v$ is slack in that case, and hence by \cref{inv:tight} we have $\dom(v) = \bs(v)$, and \cref{inv:tight} continues to hold for $v$. Note that \cref{inv:forgotten-bound} is also preserved. The bounded weight invariant (\cref{inv:weight-bound}) is maintained, as for slack vertices we can safely forget up to $(1+\epsilon)\alpha$ neighbors.
  \end{enumerate}
\end{remark}

\begin{observation}\label{obs:neighbor-levels}
  After \cref{step:sparsification}, for any $v \in \mathcal{C} \cup \mathcal{D}'$ we have $\dom(u) = \bs(u)$ for every $u \in I(v)$; for every $u \in N[v] \setminus I(v)$ we have $\dom(u) = \max \{\bs(u), k+1\}$ and $u \in S \cup T'$.
\end{observation}

\begin{lemma}\label{cl:relevant-vertices}
  After \cref{step:sparsification}, for every $v \in \mathcal{C} \cup \mathcal{D}'$ we have $N[v] \setminus \Ign{v} \subseteq S \cup T'$. Moreover,
  \[
    \sum_{v \in \mathcal{C} \cup \mathcal{D}'} |N[v] \setminus \Ign{v}|
    \le \sum_{v \in S \cup T'}\bigl(|F^{\old}(v)| + |F'(v)|\bigr).
  \]
\end{lemma}
\begin{proof}
  The first part of the lemma is implied by \cref{obs:neighbor-levels}. To show the second part, fix $v \in \mathcal{C} \cup \mathcal{D}'$ and $u \in N[v]$. Observe that after \cref{step:filter-dirty}, we have $v \in F^{\old}(u) \cup F'(u)$.

  If $u \notin S \cup T'$, then $v \in F(u)$ after \cref{step:sparsification}, so $u \in I(v)$. Otherwise, \cref{step:sparsification} sets $F(u) \gets \emptyset$, and then $u \in S \cup T'$. Each incidence $u \in N[v] \setminus I(v)$ implies that $u \in S \cup T'$ and $v \in F^{\old}(u) \cup F'(u)$. Charging this incidence to $v$ gives
  \[
    \sum_{v \in \mathcal{C} \cup \mathcal{D}'} |N[v] \setminus I(v)|
    \le \sum_{u \in S \cup T'} \bigl(|F^{\old}(u)| + |F'(u)|\bigr). \qedhere
  \]
\end{proof}

\begin{observation}\label{obs:no-outdated-below}
  After \cref{step:sparsification}, there are no outdated copies at level $k$ or below.
\end{observation}
\begin{proof}
  Observe that after \cref{step:handle-outdated}, for any vertex $v$ and any vertex $u \in \hat{N}_{\le k}[v]$ we have added $v$ to $F'(u)$. Then in \cref{step:sparsification}, we either add $v$ to $F(u)$, or we update $\view{u}{v}$ to the actual $\lev(v)$, after which $\view{u}{v} \ge k+1$.
\end{proof}

\begin{step}[Handling dirty vertices]\label{step:handle-dirty}
  For each vertex $v \in \mathcal{D}'$ find the minimum level $h \ge \max\{k+1, \bt(v)\}$ such that $\hat{\wts}(u) - \viewx{u}{v} + \wat{h} \le c_u - |F(u)| \cdot (1+\epsilon)\lambda c_u$ for all $u \in N[v] \setminus \Ign{v}$, and then set $\lev(v)$ to $h$ and update $\view{u}{v}$ to $h$ for all $u \in N[v] \setminus I(v)$. If $h = k+1$ or $h = \bt(v)$ then make $v$ active; otherwise, set $\zlev(v)$ to $k+1$. Note that $h$ can be found in time $O(\lev(v) - h + 1)$ after finding $\min_{u \in N[v] \setminus I(v)} c_u - |F(u)| \cdot (1+\epsilon)\lambda c_u - \hat{\wts}(u) + \viewx{u}{v}$, which takes time $O(|N[v] \setminus I(v)|)$.
\end{step}

\begin{remark}
  Recall that $\mathcal{D}'$ consists of dirty vertices for which we failed to find a tight neighbor $u \in N[v] \setminus I(v)$ such that $\dom(u) \ge k+1$. Thus, the lazy level invariant (\cref{inv:zlev}) might be violated for $v \in \mathcal{D}'$. We try to fix it by decreasing $\lev(v)$ until some vertex $u \in N[v] \setminus \Ign{v}$ becomes tight. Then we set $\lev(v)$ to $h$ and update $\view{u}{v}$ to $h$ for each $u \in N[v] \setminus \Ign{v}$. Note that because of the way we select $h$, the bounded weight invariant (\cref{inv:weight-bound}) is maintained.

  Setting $\lev(v)$ and all $\view{u}{v}$ for $u \in N[v] \setminus I(v)$ to $h$ satisfies the level invariant and the view invariant (\cref{inv:level,inv:view}), since by \cref{obs:neighbor-levels}, we have $\dom(u) \le \max{\bt(v),k+1}$ for all $u \in N[v]$.

  If $h > \max\{k+1, \bt(v)\}$, then we argue that at least one of the vertices in $N[v] \setminus \Ign{v}$ is tight. Indeed, the argument here is the same as the one we used for the $\ins$ subroutine. We have $\hat{\wts}(u) - \view{u}{v} + \wat{h} > \gamma c_u / (1+\epsilon)$ for at least one $u \in N[v] \setminus \Ign{v}$. By \cref{obs:neighbor-levels}, we have $\dom(u) \ge k+1$ for such neighbor $u$, and so setting $\zlev(v)$ to $k+1$ restores the lazy level invariant (\cref{inv:zlev}). If $h = \bt(v)$, then $v$ becomes active. Otherwise, if $h \ne \bt(v)$ and $h = k+1$, we also make $v$ active, but with the following caveat. If $N[v] \setminus I(v) \ne \emptyset$, then $\Dom(v) = k+1$, according to \cref{obs:neighbor-levels}. However, it is possible that $N[v] \setminus I(v) = \emptyset$, and hence $\Dom(v) < k+1$. As we argue later, such vertices will be placed at level $\bt(v)$ after \cref{step:water}; thus $v$ will again be active and in a correct state. We also note that it is possible that $N[v] \setminus I(v) = \emptyset$. But in that case $v$ will reach $\bt(v)$ in the subsequent call to $\water$ and thus join $H$.

  Therefore, after \cref{step:handle-dirty}, all passive vertices have $\zlev(v) \ge k+1$, and the lazy level invariant (\cref{inv:zlev}) is satisfied for them.
\end{remark}

\begin{step}\label{step:pre-water}
  Let $D' \subseteq S \cup T'$ be the set of slack vertices $v$ such that $\dom(v) > \bs(v)$, and let $V' \subseteq \mathcal{C} \cup \mathcal{D}'$ be the set of vertices $v$ such that $\lev(v) > \bt(v)$ and there is no tight neighbor among $N[v] \setminus I(v)$. For every $v \in D'$ set $\dom(v) \gets k$. For every $v \in V'$ set $\lev(v) \gets k$ and update $\view{u}{v}$ to $k$ for all $u \in N[v] \setminus I(v)$.
\end{step}

\begin{remark}
  After this step, the tightness invariant (\cref{inv:tight}) might be violated only for vertices in $D'$ and they are at dominating level $k$. All vertices with $\Dom(v) \le k$ are in $V'$.
\end{remark}

\begin{step}\label{step:water}
Invoke $\water(k, D', V')$.
\end{step}

\begin{remark}
  In this step we fix \cref{inv:tight} for vertices $D'$ using the $\water$ subroutine from \cref{lm:water-filling-full}. The original subroutine from \cite{bhattacharya2019new} is intended for the MSC problem. The adaptation for the MDS problem is straightforward, and we describe the necessary modifications to make it work for our algorithm. For the description of the original subroutine, refer to \cite{bhattacharya2019new}.
\end{remark}

\begin{lemma}[From \cite{bhattacharya2019new}, adapted for MDS]\label{lm:water-filling-full}
  There exists a subroutine $\water$ with the following properties.\footnote{In \cite{bhattacharya2019new}, this subroutine is called $\textrm{FixLevel}$.} The subroutine $\water(k, D', V')$ takes as input two collections of vertices $D'$ and $V'$ such that $\dom(v)=k$ for every $v \in D'$ and $\Dom(u)=k$ for every $u \in V'$. Moreover, each vertex $v \in D'$ satisfies $\wts(v) \le c_v$. The subroutine updates the dominating levels of vertices in $D'$ and the dominated levels of vertices in $V'$ so that: (1) $\wts(v) \le c_v$ for every $v \in D'$, and (2) if $\dom(v) > \bs(v)$, then $\wts(v) > c_v/(1+\epsilon)$ for every $v \in D'$. The subroutine runs in time $O\!\left(k + |D'| + |V'| + \sum_{u \in V'} |N[u] \cap D'|\right)$.
\end{lemma}

\paragraph{Adapting the $\water$ subroutine to our algorithm}
The original $\water$ subroutine from \cite{bhattacharya2019new} works for the MSC problem, and takes as input a collection of sets and elements. In our case, vertices from $V'$ translate to elements and vertices from $D'$ translate to sets.
One can observe by the algorithm of \cite{bhattacharya2019new} that the runtime would be $O(k + |D'| + |V'| + \sum_{v \in V'} |N[v] \setminus I(v)|)$. In addition, during the $\water$ subroutine, we enforce $\view{u}{v} = \lev(v)$ for all $u \in N[v] \setminus I(v)$. Thus, whenever we update $\lev(v)$ (which is equal to $\Dom(v)$), we also update $\view{u}{v}$ for all $u \in N[v] \setminus I(v)$. Note that this does not affect the total runtime.

The $\water$ subroutine from \cite{bhattacharya2019new} does not enforce $\lev(v) \ge \bt(v)$, which could break \cref{inv:level}. We enforce this by adding a fake neighbor of cost $(1+\epsilon)\lambda\tau_v$ for every $v \in V'$. Thus, whenever $v$ reaches level $\bt(v)$, this fake neighbor becomes tight and so $v$ becomes ``frozen'' (for the description of ``frozen'' elements, refer to \cite{bhattacharya2019new}).

In the end, we set $\dom(v) \gets \max\{\dom(v), \bs(v)\}$ for every $v \in D'$ to restore the requirement $\dom(v) \ge \bs(v)$ from \cref{inv:tight}. Note that this is safe to do, since $\lev(v) \ge \bt(v)$ for all $v \in V'$.

\begin{remark}
At the beginning of \cref{step:water}, all vertices $v \in V'$ are active and have $\lev(v) = k$, but some of them might be in an incorrect state, since it might be $N[v] \setminus I(v) = \emptyset$. However, the $\water$ subroutine will place such vertices at level $\bt(v)$, and so $\lev(v) = \Dom(v) = \bt(v)$ for them.
\end{remark}

\begin{step}\label{step:return}
  RETURN.
\end{step}

\section{Amortized runtime analysis}\label{sec:analysis}
Our analysis uses a potential function composed of several terms. The total potential $\Phi$ is defined as the sum of the following potential functions.
\begin{description}
  \item[$\phi$-potential.] For each index $0 \le i \le L$, define $\beta_i = \frac{18}{\epsilon^2}\cdot\frac{(1+\epsilon)^{i+1}}{\lambda}$. The $\phi$-potential of a vertex $v$ is $\Phi^{\phi}(v) = \phi(v)\cdot \beta_{\dom(v)}$. The $\phi$-potential of level $i$ is $\Phi^{\phi}_i = \sum_{v \in D_i} \Phi^{\phi}(v)$.

  \item[$\delta$-potential.] The $\delta$-potential of level $i$ is $\Phi^{\delta}_i = (\alpha+1)\cdot \delta_i \cdot 2\beta_i$.

  \item[Passive potential.] For a passive vertex $v$, define $\Phi^{\passive}(v) = \frac{18\alpha}{\epsilon}\cdot(\lev(v)-\zlev(v))$. For an active vertex $v$, define $\Phi^{\passive}(v)=0$.

  \item[Forget potential.] The forget potential of a vertex $v$ is $\Phi^{\forget}(v) = -\frac{1}{2}|\Forg{v}|$.

  \item[Level potential.] For each level $i$, define $\Phi^{\level}_i = L$ if $\phi_i \neq 0$ or $\delta_i \neq 0$, and define $\Phi^{\level}_i = 0$ otherwise.

  \item[Incoming potential.]
    Consider some orientation of the edges of $G$, which yields an oriented graph $\overline{G}$. We say that this orientation is a $d$-orientation if the out-degree of each vertex is bounded by $d$. Since we add and delete edges from $G$, $\overline{G}$ also changes. Moreover, some edges might change their orientation. The number of edges that changed their orientation is called the number of reorientations. We use the following theorem from \cite{brodal1999dynamic}
    \begin{theorem}\label{thm:reorientation}
      Given an arboricity-$\alpha$-preserving sequence of edge insertions and deletions on an initially empty graph, for any $d > \alpha$ there exists a sequence of $d$-orientations such that
      \begin{enumerate}
      \item for each edge insertion there are no edge reorientations,
      \item for each edge deletion there are at most $\ceil{\log_{d/\alpha} |V|}$ reorientations.
      \end{enumerate}
    \end{theorem}
    Let $\overline{G}$ be the oriented graph from such a sequence of $(1+\epsilon/2)\alpha$-orientations. Then the incoming potential of a vertex $v$ is defined as $\Phi^{in}(v) = \frac{3}{\epsilon} | \Forg{v} \cap \overline{N}^{in}[v] |$.
\end{description}

For each type of potential, we also define its total value by summing over all vertices or all levels, as appropriate. For example,
\[
  \Phi^{\passive} = \sum_{v \in V} \Phi^{\passive}(v)
  \qquad\text{and}\qquad
  \Phi^{\delta} = \sum_{i=0}^L \Phi^{\delta}_i.
\]
Similarly, we define $\Phi^{\forget}$, $\Phi^{\level}$ and $\Phi^{in}$. The total potential is then
\[
  \Phi = \Phi^{\phi} + \Phi^{\delta} + \Phi^{\passive} + \Phi^{\forget} + \Phi^{\level} + \Phi^{in}.
\]

Each unit decrease in potential will be used to pay for a constant amount of runtime.

\subsection{Insertion}\label{sec:insertion-analysis}
From the description of the $\ins$ subroutine, the actual runtime cost of insertion is $O(L) = O(\log_{1+\epsilon} Cn) = O(\frac{\log n}{\epsilon} + \frac{\log C}{\epsilon})$. Next, we analyze the potential increase. Without loss of generality, consider one endpoint, say $v$.

Let $\lev^{\old}(v)$ be the value of $\lev(v)$ before the insertion. Notice that we might have created at most one outdated copy, and this copy is at level $\lev^{\old}(v)$. This, as a result, increases $\delta_{\lev^{\old}(v)}$ by $\wat{\lev^{\old}(v)}$. Therefore, the $\delta$-potential increases by at most
\[
  (\alpha+1) \cdot \wat{\lev^{\old}(v)} \cdot 2\frac{18}{\epsilon^2}\cdot \frac{(1+\epsilon)^{\lev^{\old}(v)+1}}{\lambda} = \frac{36(1+\epsilon)(\alpha+1)}{\epsilon^2}.
\]
The passive potential might increase by at most $\frac{18\alpha L}{\epsilon} = O(\frac{\alpha \log n}{\epsilon^2} + \frac{\alpha \log C}{\epsilon^2})$, since the gap $\lev(v) - \zlev(v)$ can increase by at most $L$. The level potential might increase by at most $L = O(\frac{\log n}{\epsilon} + \frac{\log C}{\epsilon})$, since $\delta_{\lev^{\old}(v)}$ might have become nonzero. The $\phi$-potential does not increase, since dead weights do not change. Observe that $F(v)$ does not change, so the forget potential does not change. According to \cref{thm:reorientation}, there are no reorientations, so the incoming potential does not change. We conclude the analysis with the following claim.

\begin{claim}\label{thm:insert-runtime}
  The amortized runtime of an insertion is $O\!\left(\frac{\alpha \log (Cn)}{\epsilon^2}\right)$.
\end{claim}

\subsection{Deletion}\label{sec:deletion-analysis}
From the description of the subroutine, the actual runtime of a deletion is $O(\log n)$ (recall that we update $\tau_v$ and $\tau_u$, which takes time $O(\log n)$). Next, we analyze the potential increase. As in the insertion analysis, it's enough to consider one endpoint, say $v$.

The increase in passive potential is at most $\frac{18\alpha}{\epsilon}L = O(\frac{\alpha \log n}{\epsilon^2} + \frac{\alpha \log C}{\epsilon^2})$, since we might decrease $\zlev(v)$. To analyze the increase in the $\phi$-potential, consider two cases. If $\dom(v) > \bs(v)$, then we increase $\phi(v)$ by $\wat{\view{v}{u}} \le \wat{\dom(v)}$ by \Cref{inv:view}, and so the increase in $\Phi^{\phi}(v)$ is at most
\[
  \wat{\dom(v)} \cdot \beta_{\dom(v)} = \wat{\dom(v)} \cdot \frac{18}{\epsilon^2}\cdot\frac{(1+\epsilon)^{\dom(v)+1}}{\lambda} = \frac{18(1+\epsilon)}{\epsilon^2}.
\]
Otherwise, if $\dom(v) = \bs(v)$, then the increase in $\Phi^{\phi}(v)$ is
\[
  \gamma c_v \cdot \beta_{\bs(v)} \le (1+\epsilon)^{-\bs(v)} \cdot \frac{18}{\epsilon^2}\cdot\frac{(1+\epsilon)^{\bs(v)+1}}{\lambda} = \frac{18(1+\epsilon)(3\alpha+1)}{\epsilon^2}.
\]
If $u \in F(v)$, then we remove $u$ from $F(v)$, so the increase in $\Phi^{\forget}(v)$ is at most $\frac{1}{2}$. The increase in the level potential is at most $L$, since $\phi_{\dom(v)}$ might have become nonzero. According to \cref{thm:reorientation}, there are at most $\left\lceil \log_{1+\epsilon/2} n \right\rceil = O(\frac{\log n}{\epsilon})$ reorientations. Therefore, the increase in the incoming potential is $O(\frac{\log n}{\epsilon^2})$. We conclude the analysis with the following claim.

\begin{claim}\label{thm:delete-runtime}
  The amortized runtime cost of deletion is $O\!\left(\frac{\alpha \log (Cn)}{\epsilon^2}\right)$.
\end{claim}

\subsection{Rebuild}\label{sec:rebuild-analysis}
Let us denote by the superscript ``old'' the values of the variables at the beginning of the execution of $\reset(k)$, i.e., $\lev^{\old}(v)$ is the value of $\lev(v)$ for a vertex $v$ at the beginning of the call. We use the notation $\Delta\Phi = \Phi - \Phi^{\old}$; we extend this notation for all potential functions.

Let us first bound the decrease in the $\delta$ and the $\phi$
potential in terms of the sizes of $V_{\le k}$, $T_{\le k}$ and
$Q_{\le k}$. Our proof is similar to the one in \cite{bhattacharya2021dynamic}.

\begin{lemma}\label{lm:down}
  We have the following lower bound on the $\phi$-potential and the $\delta$-potential at the beginning of $\reset(k)$:
  \[
  \Phi^{\delta}_{\le k} + \Phi^{\phi}_{\le k}
  \ge 18(\alpha+1) |T_{\le k}|  + \frac{18(\alpha+1)}{\epsilon}|V_{\le k}| + \frac{18(\alpha+1)}{\epsilon^2} |Q_{\le k}|.
  \]
\end{lemma}
\begin{proof}
  By definition, $k$ is the minimum index such that:
  \begin{equation}\label{first}
    (\alpha+1)\delta_{\le k} + \phi_{\le k} ~>~ \epsilon (\xi c(T_{\le k}) + (\alpha+1)\sum_{v \in V_{\le k}} x_v),
  \end{equation}
  \begin{equation}\label{second}
    (\alpha+1)\delta_{\le i} + \phi_{\le i} ~\leq~ \epsilon (\xi c(T_{\le i}) + (\alpha+1)\sum_{v \in V_{\le i}} x_v), \forall 0 \le i < k.
  \end{equation}
  Therefore, we have:
  \begin{equation}\label{eq:1}
    \begin{aligned}
      &\frac{1}{2}\Phi^{\delta}_{\le k} + \Phi^{\phi}_{\le k} =
        \sum_{i = 0}^{k}\beta_i\cdot \left((\alpha+1)\delta_i + \phi_i\right) \\
      &= \beta_k\cdot \left((\alpha+1) \delta_{\le k} + \phi_{\le k}\right) - \sum_{i =0}^{k-1}(\beta_{i+1} - \beta_i)\cdot \left((\alpha+1)\delta_{\leq i} + \phi_{\le i}\right)\\
      &>\beta_k\cdot \epsilon\cdot \left(\xi c(T_{\le k}) + (\alpha+1)\sum_{v \in V_{\le k}}x_v\right) - \sum_{i=0}^{k-1}(\beta_{i+1} - \beta_i)\cdot \epsilon\cdot  \left(\xi c(T_{\le i}) + (\alpha+1)\sum_{v \in V_{\le i}}x_v\right)\\
      &=\epsilon \cdot \sum_{i = 0}^k \beta_i \cdot \xi c(T_i) + \epsilon (\alpha+1) \sum_{i =0}^k \beta_i \cdot \sum_{v \in V_i}x_v
        \mbox{~~(by Abel transformation)}
      \\
      &\geq \epsilon \sum_{i = 0}^k \sum_{v \in T_i} \xi c_v \cdot \frac{18}{\epsilon^2}\cdot\frac{(1+\epsilon)^{i+1}}{\lambda} + \epsilon (\alpha+1)\sum_{i =0}^k \sum_{v\in V_{i}} x_v \cdot \frac{18}{\epsilon^2} \cdot\frac{(1+\epsilon)^{i+1}}{\lambda}\\
      &\geq (\alpha+1) \sum_{i = 0}^k \sum_{v \in T_i} (1+\epsilon)^{-\bs(v)-1}\cdot \frac{18}{\epsilon} \cdot (1+\epsilon)^{i+1}  + (\alpha+1) \sum_{i = 0}^k\sum_{v\in V_i} \lambda(1+\epsilon)^{-i} \cdot \frac{18}{\epsilon} \cdot \frac{(1+\epsilon)^{i+1}}{\lambda} \\
      &\geq \sum_{v \in T_{\leq k}} 18(\alpha+1) \cdot (\dom(v) - \bs(v)+1) + \frac{18(\alpha+1)}{\epsilon}\sz{V_{\le k}}\\
      &\ge 18(\alpha+1)|T_{\le k}| + \frac{18(\alpha+1)}{\epsilon}|V_{\le k}|,
    \end{aligned}
  \end{equation}
where the first inequality follows from \Cref{first} and \Cref{second}; the last two inequalities hold since $c_v \ge (1 + \epsilon)^{-\bs(v) - 1}$ and $(1 + \epsilon)^x \ge 1 + \epsilon x \ge \epsilon (1 + x)$, and since $\dom(v) \ge \bs(v)$ by \cref{inv:tight}.
  Next, recall that $\delta_i = \wat{i} \cdot |Q_i|$. Thus,
  \begin{equation}\label{eq:2}
    \frac{1}{2}\Phi^{\delta}_i = (\alpha+1) \cdot \delta_i \cdot \beta_i= \wat{i} \cdot \frac{18}{\epsilon^2}\cdot\frac{(1+\epsilon)^{i+1}}{\lambda} |Q_i| \ge \frac{18(\alpha+1)}{\epsilon^2} |Q_i|.
  \end{equation}
  Summing \cref{eq:2} over $0 \le i \le k$ and combining with \cref{eq:1}, we get the claimed bound.
\end{proof}

Recall that the $\reset(k)$ subroutine eliminates all dead weights for all vertices with dominating level at most $k$ and all outdated copies at level at most $k$. Therefore, at the end of the call we have $\Phi^{\phi}_{\le k} + \Phi^{\delta}_{\le k} = 0$. Moreover, we do not increase $\Phi^{\phi}_{\ge k+1}+\Phi^{\delta}_{\ge k+1}$. Next, we bound the increase in $\Phi^{\passive}$.

\begin{observation}\label{obs:passive-level}
  Each dirty vertex $v \in \mathcal{D}$ either becomes active (i.e.\ $\lev(v) = \zlev(v)$) or satisfies $\lev^{\old}(v) \ge \lev(v)$ and $\zlev(v) = k+1$. Each passive vertex $v \in \mathcal{C}$ becomes active. Any other passive vertex $v$ is unaffected by $\reset(k)$, i.e.\ $\lev^{\old}(v) = \lev(v)$ and $\zlev^{\old}(v) = \zlev(v)$.
\end{observation}

\begin{lemma}\label{lm:passive-decrease}
  The change in passive potential satisfies $\Delta\Phi^{\passive} \le - \frac{18\alpha}{\epsilon}|P_{\le k}|$. Moreover, the decrease in $\Phi^{\passive}(v)$ for $v \in \mathcal{D}'$ is enough to cover any $O(\lev^{\old}(v) - \lev(v) + 1)$ runtime cost.
\end{lemma}
\begin{proof}
  By \Cref{obs:passive-level}, each passive vertex $v \in \mathcal{C}$ becomes clean. Each vertex $v \in \mathcal{D}$ either becomes clean, or $\zlev(v)$ becomes $k+1$. Additionally, we have $\lev^{\old}(v) \ge \lev(v)$ for $v \in \mathcal{D}$. Therefore, the gap $\lev(v) - \zlev(v)$ decreases by at least one for every $v \in P_{\le k}$ (observe that all such vertices are either clean or dirty), and hence $\Phi^{\passive}(v)$ decreases by at least $\frac{18\alpha}{\epsilon}$. More specifically, for vertices $v \in \mathcal{D}'$, we can lower-bound this decrease by $\frac{18\alpha}{\epsilon} (\lev^{\old}(v) - \lev(v) + 1)$. For other passive vertices, their $\lev(v)$ and $\zlev(v)$ remain unchanged, and hence the passive potential remains the same.
\end{proof}

Let $f_v = |F^{\old}(v)|$ and $f'_v = |F'(v)|$. Next, we bound the potential decrease in terms of $f_v$ and $f'_v$.

\begin{lemma}\label{lm:F-bound} $\sum_{v\in V} \big|F'(v)\cap \overline{N}^{in}[v]\big|
  \le 3\alpha\,\big|Q_{\le k}\cup P_{\le k} \cup V_{\le k}\big|$.
\end{lemma}
\begin{proof}
  \[
    \sum_{v\in V} \big|F'(v)\cap \overline{N}^{in}[v]\big|
    =\sum_{v\in V}\sum_{u\in F'(v)} \mathbf{1}[u\in\overline{N}^{in}[v]]
    \le \sum_{v\in V}\sum_{u\in Q_{\le k}\cup P_{\le k}\cup V_{\le k}} \mathbf{1}[u\in\overline{N}^{in}[v]].
  \]
  Swapping the sums yields
  \[
    =\sum_{u\in Q_{\le k}\cup P_{\le k}\cup V_{\le k}} \big|\overline{N}^{out}[u]\big|.
  \]
  We have $|\overline{N}^{out}(u)|\le(1+\epsilon/2)\alpha$ for all $u$, so
  \[
    \big|\overline{N}^{out}[u]\big|\le(1+\epsilon/2)\alpha+1\le 3\alpha,
  \]
  using $\epsilon\le1$ and $\alpha\ge1$. Summing over $u$ gives the claim.
\end{proof}

\begin{lemma}\label{lm:large-bound}
  For any vertex $v$ such that $f_v+f'_v\ge(1+\epsilon)\alpha$ we have
  \[
    f_v+f'_v\le\frac{3}{\epsilon}\Bigl(\bigl|\Forg{v}\cap\overline{N}^{in}[v]\bigr|+\bigl|F'(v)\cap\overline{N}^{in}[v]\bigr|\Bigr).
  \]
\end{lemma}
\begin{proof}
  Let $v$ satisfy $f_v+f'_v\ge(1+\epsilon)\alpha$. The out-degree bound $(1+\epsilon/2)\alpha$ gives
  \[
    \bigl|\Forg{v}\cap\overline{N}^{out}(v)\bigr|+\bigl|F'(v)\cap\overline{N}^{out}(v)\bigr|
    \le(1+\tfrac{\epsilon}{2})\alpha
    \le\frac{1+\epsilon/2}{1+\epsilon}\bigl(f_v+f'_v\bigr).
  \]
  Hence
  \[
    \bigl|\Forg{v}\cap\overline{N}^{in}[v]\bigr|+\bigl|F'(v)\cap\overline{N}^{in}[v]\bigr|
    \ge\Bigl(1-\frac{1+\epsilon/2}{1+\epsilon}\Bigr)\bigl(f_v+f'_v\bigr)
    =\frac{\epsilon}{2(1+\epsilon)}\bigl(f_v+f'_v\bigr).
  \]
  Therefore $f_v+f'_v\le\frac{2(1+\epsilon)}{\epsilon}\bigl(|\Forg{v}\cap\overline{N}^{in}[v]|+|F'(v)\cap\overline{N}^{in}[v]|\bigr)$, and since $\epsilon\le\frac12$ we get the claimed bound.
\end{proof}

To bound the increase in $\Phi^{\forget}$ and $\Phi^{in}$, we make the following observations.

\begin{observation}\label{obs:F}
  For any $v \in S \cup T'$ we have $F(v) = \emptyset$. For any $v \notin S \cup T'$, we have $F(v) = F^{\old}(v) \sqcup F'(v)$.
\end{observation}

\begin{observation}\label{obs:S}
  The set $S$ consists of all vertices $v$ such that $f_v + f'_v > (1+\epsilon)\alpha$.
\end{observation}

\begin{lemma}\label{lm:forget-increase}
  The increase in the forget potential $\Phi^{\forget}$ is at most
  \[
    \alpha |T^{\old}_{\le k}| + \frac{1}{2}\sum_{v \in S} f_v
    - \frac{1}{2}\sum_{v \notin S \cup T'} f'_v.
  \]
\end{lemma}
\begin{proof}
  By \cref{obs:F}, for every vertex $v \notin S \cup T'$, the forget potential $\Phi^{\forget}(v)$ decreases by exactly $\frac{1}{2}f'_v$. On the other hand, for every vertex $v \in S \cup T'$, the forget potential $\Phi^{\forget}(v)$ increases by $\frac{1}{2}f_v$. By \cref{inv:forgotten-bound}, we have $f_v \le (1+\epsilon)\alpha \le 2\alpha$. Hence,
  \[
    \sum_{v \in S \cup T'} f_v \le \sum_{v \in S} f_v + 2\alpha |T'|.
  \]
  To complete the proof, recall that $T' = T^{\old}_{\le k}$.
\end{proof}

\begin{lemma}\label{lm:total-decrease}
  The decrease in $\Phi^{\delta} + \Phi^{\phi} + \Phi^{\passive} + \Phi^{in}+ \Phi^{\forget}$ is
  \[\Omega\Big(\alpha |T^{\old}_{\le k}| + \frac{\alpha}{\epsilon}|V^{\old}_{\le k}| + \frac{\alpha}{\epsilon}|Q^{\old}_{\le k}| + \frac{\alpha}{\epsilon}|P^{\old}_{\le k}| + \sum_{v \in S} (f_v + f'_v) + \sum_{v \notin S \cup T'} f'_v\Big). \]
\end{lemma}
\begin{proof}
  By \Cref{obs:F},
  \[
    \Delta\Phi^{in}(v)=
    \begin{cases}
      -\dfrac{3}{\epsilon}\big|F^{\old}(v)\cap\overline{N}^{in}[v]\big| & v\in S \cup T',\\[6pt]
      \dfrac{3}{\epsilon}\big|F'(v)\cap\overline{N}^{in}[v]\big| & v\notin S \cup T',
    \end{cases}
  \]

  hence
  \begin{equation}\label{eq:3}
    \Delta\Phi^{in}
    = \sum_{v\notin S \cup T'}\frac{3}{\epsilon}\big|F'(v)\cap\overline{N}^{in}[v]\big|
    -\sum_{v\in S \cup T'}\frac{3}{\epsilon}\big|F^{\old}(v)\cap\overline{N}^{in}[v]\big|.
  \end{equation}
  From \Cref{lm:down,lm:passive-decrease},
  \[
    \Delta\Phi^{\phi}+\Delta\Phi^{\delta}\le -\frac{18(\alpha+1)}{\epsilon}|V_{\le k}| - \frac{18(\alpha+1)}{\epsilon^2}|Q_{\le k}|,
    \qquad
    \Delta\Phi^{\passive}\le -\frac{18\alpha}{\epsilon}|P_{\le k}|,
  \]
  so
  \[
    \frac{1}{2}\Delta\Phi^{\phi} + \frac{1}{2}\Delta\Phi^{\delta} + \frac{1}{2}\Delta\Phi^{\passive}
    \le -\frac{9\alpha}{\epsilon}|V_{\le k}\cup P_{\le k}\cup Q_{\le k}|.
  \]
  Applying \Cref{lm:F-bound} yields
  \begin{equation}\label{eq:4}
    \frac{1}{2}\Delta\Phi^{\phi} + \frac{1}{2}\Delta\Phi^{\delta} + \frac{1}{2}\Delta\Phi^{\passive}
    \le -\sum_{v\in V}\frac{3}{\epsilon}\big|F'(v)\cap\overline{N}^{in}[v]\big|.
  \end{equation}
  Summing \cref{eq:3,eq:4} we obtain
  \[
    \frac{1}{2}\Delta\Phi^{\phi} + \frac{1}{2}\Delta\Phi^{\delta} + \frac{1}{2}\Delta\Phi^{\passive} + \Delta\Phi^{in}
    \le -\sum_{v\in S \cup T'}\frac{3}{\epsilon}\Big(\big|F^{\old}(v)\cap\overline{N}^{in}[v]\big|+\big|F'(v)\cap\overline{N}^{in}[v]\big|\Big).
  \]
  By \cref{obs:S}, for every $v \in S$ we have $f_v + f'_v > (1+\epsilon)\alpha$. Thus, applying \cref{lm:large-bound} yields
  \[
    \frac{1}{2}\Delta\Phi^{\phi} + \frac{1}{2}\Delta\Phi^{\delta} + \frac{1}{2}\Delta\Phi^{\passive} + \Delta\Phi^{in}
    \le -\sum_{v\in S}(f_v + f'_v).
  \]
  Combining this bound with \cref{lm:down,lm:passive-decrease,lm:forget-increase}, we obtain the claimed bound.
\end{proof}

Next, we analyze the total runtime of the $\reset(k)$ subroutine, excluding the $O(\lev^{\old}(v)-\lev(v)+1)$ cost incurred when processing vertices $v \in \mathcal{D}'$ during \cref{step:handle-dirty}. By \cref{lm:passive-decrease}, this part of the runtime can be charged to the decrease in the passive potential.

\begin{lemma}\label{lm:rebuild-runtime-cost}
  The runtime of $\reset(k)$, excluding the $O(\lev^{\old}(v)-\lev(v)+1)$ cost of searching $h$ during \cref{step:handle-dirty}, is
  \[
    O\!\left(
      k + \alpha |T^{\old}_{\le k}| + \frac{\alpha}{\epsilon}|V^{\old}_{\le k}| + \frac{\alpha}{\epsilon}|Q^{\old}_{\le k}| + \frac{\alpha}{\epsilon}|P^{\old}_{\le k}| + \sum_{v \in S} (f_v+f'_v) + \sum_{v \notin S \cup T'} f'_v
    \right).
  \]
\end{lemma}
\begin{proof}
  We bound the runtime of $\reset(k)$ step by step.

  \Cref{step:init} takes time $O(|V^{\old}_{\le k}| + |P^{\old}_{\le k}| + |T^{\old}_{\le k}| + k)$, and \cref{step:remove-dead-weight} takes time $O(k + |T^{\old}_{\le k}|)$. \Cref{step:handle-outdated,step:handle-clean,step:filter-dirty} together take time $O(k + |V^{\old}_{\le k}| + |P^{\old}_{\le k}| + \sum_{v \in V} f'_v)$.

  We next consider \cref{step:sparsification}. For each vertex $v \in S \cup T'$, the runtime is $O(f_v + f'_v)$, whereas for each vertex $v \notin S \cup T_{\le k}$, it is $O(f'_v)$.

  For \cref{step:handle-dirty}, let $v \in \mathcal{D}'$. Ignoring the cost of finding $h$, processing $v$ takes time $O\left(|N[v] \setminus \Ign{v}|\right)$. Hence, by \cref{cl:relevant-vertices},
  \[
    \sum_{v \in \mathcal{D}'} O(|N[v] \setminus \Ign{v}|)
    = O\!\left(\sum_{v \in \mathcal{D}'} |N[v] \setminus \Ign{v}|\right)
    = O\!\left(\sum_{v \in S \cup T'} (f_v + f'_v)\right).
  \]

  By \cref{lm:water-filling-full}, \cref{step:pre-water,step:water} run in time
  $O(k + |V^{\old}_{\le k}| + |P^{\old}_{\le k}| + |T^{\old}_{\le k}| + \sum_{v \in S \cup T'} (f_v + f'_v))$.

  Finally, by \cref{obs:S} we have $f_v + f'_v \le (1+\epsilon)\alpha$ for every $v \in T' \setminus S$. Therefore,
  \[
    \sum_{v \in S \cup T'} (f_v + f'_v) = \sum_{v \in S} (f_v+f'_v) + O(\alpha |T^{\old}_{\le k}|).
  \]

  Combining all of the above bounds yields the claimed running time.
\end{proof}

Therefore, according to \cref{lm:total-decrease,lm:rebuild-runtime-cost}, the runtime can be charged to the decrease in potential, except the $O(k)$ part. The $O(k)$ part can be charged to the decrease in the level potential, since for the $\reset(k)$ to be triggered, there must be at least one level $i \le k$ such that $\delta_i + \phi_i \ne 0$. We conclude our analysis with the following claim.
\begin{claim}\label{thm:rebuild-runtime}
  The runtime cost of $\reset(k)$ can be charged to the decrease in potential.
\end{claim}

\subsection{Total runtime}\label{sec:total-runtime}
During the preprocessing we initialize all the data structures, which takes time $O(nL) = O(\frac{n \log (Cn)}{\epsilon})$. To bound the amortized update time, consider any update sequence of length $\Gamma$. Let $\Phi^{\text{init}}$ be the value of the total potential $\Phi$ at the beginning of the sequence, and $\Phi^{\text{end}}$ be the value at the end of this sequence. Observe that $\Phi^{\text{init}} = 0$. We claim that the value of $\Phi^{\text{end}}$ is at least $-(\Gamma+n)$. Indeed, observe that all potential functions are non-negative, except the forget potential. We have
\[
  \Phi^{\forget} = -\sum_{v \in V} \frac{1}{2}|F(v)| \ge - \sum_{v \in V}\frac{1}{2}(\deg(v)+1) \ge -(m+n) \ge -(\Gamma+n),
\]
where the last inequality is since the initial graph has no edges, and so $m \le \Gamma$. According to \cref{thm:insert-runtime,thm:delete-runtime,thm:rebuild-runtime}, the total update time is bounded asymptotically by
\[
  \Gamma \cdot \frac{\alpha \log (Cn)}{\epsilon^2} + \Phi^{\text{init}} - \Phi^{\text{end}} \le \Gamma \cdot \left(\frac{\alpha \log (Cn)}{\epsilon^2}+1\right) + n.
\]
Observe that the $n$ term is subsumed by the runtime of the preprocessing. Thus we conclude the analysis with the following theorem.

\begin{theorem}\label{thm:update-time}
  There exists a deterministic algorithm for the minimum dominating set problem with approximation ratio $(1+5\epsilon) \cdot 2(3\alpha+1)$ and amortized update time $O\!\left(\frac{\alpha \log (Cn)}{\epsilon^2}\right)$ and preprocessing time $O\! \left( \frac{n \log (Cn)}{\epsilon} \right)$.
\end{theorem}

\section{Conclusions and Future Work}\label{sec:conclusion}
In this work, we presented a dynamic algorithm for the minimum dominating set problem in bounded-arboricity graphs, with amortized update time $O(\alpha \log (Cn))$ and approximation ratio $O(\alpha)$. Our work leaves several natural open questions.

Can the approximation factor be improved from $O(\alpha)$ to the essentially optimal bound of $\alpha+1$ (or even $\alpha$)? This was achieved in the static setting \cite{sun2021}, and it would be interesting to obtain the same guarantee dynamically. Such a result would be essentially optimal, since achieving an $(\alpha-\epsilon)$-approximation in polynomial time is impossible under the Unique Games Conjecture \cite{BU17}.

Our update time guarantee is amortized.
The work of Solomon et al.\ \cite{solomon2024lossless} (improving over \cite{bhattacharya2021dynamic})
gives an $(1+\eps)\ln \Delta$-approximation algorithm
with a worst-case update time of $O\left(\frac{\Delta \log n}{\epsilon^2}\right)$; note that there is no dependency on $\log C$ in their time bound, as discussed below.
The following question for bounded arboricity graphs naturally arises: is it possible to obtain an approximation ratio of $O(\alpha)$ with a \emph{worst-case} update time of $O\!\left(\alpha \log(Cn)\right)$?

Finally, our approach incurs a factor of $\log(Cn)$ in the update time. This raises the following two questions. First, is it possible to eliminate the dependency on the aspect ratio $C$ in our update-time bound? As was shown in \cite{solomon2024lossless}, one can obtain update time $O\left(\frac{\Delta \log n}{\epsilon^2}\right)$ with approximation ratio $(1+\eps)\ln \Delta$. Second, is it possible to remove or reduce the $\log n$ factor from the update time?

Combining all these questions together, is it possible to obtain a dynamic algorithm with a worst-case update time of $O(\alpha)$ for approximation ratio that approaches $\alpha$?

\section*{Acknowledgements}
Shay Solomon is funded by the European Union (ERC, DynOpt, 101043159). Views and opinions expressed are however those of the author(s) only and do not necessarily reflect those of the European Union or the European Research Council. Neither the European Union nor the granting authority can be held responsible for them. Shay Solomon is also funded by a grant from the United States-Israel Binational Science Foundation (BSF), Jerusalem, Israel, and the United States National Science Foundation (NSF).
Anton Bukov is funded by the European Union (ERC, DynOpt, 101043159).

\appendix
\section{Glossary}
\label{sec:glossary}

\setlength{\tabcolsep}{5pt}
\begin{longtable}{| l | p{13.5cm} |}
  \caption{Glossary of the main notation used in this paper.}\label{table:glossary}\\
  \hline
  \textbf{Notation} & \textbf{Definition} \\ \hline
  \endfirsthead

  \hline
  \textbf{Notation} & \textbf{Definition} \\ \hline
  \endhead

  \hline
  \endfoot

  $n$ & Number of vertices, $n = |V|$. \\ \hline
  $c_v$ & Cost of vertex $v$, with $c_v \in [\frac{1}{C}, 1]$. \\ \hline
  $C$ & Cost aspect-ratio bound. \\ \hline
  $\epsilon$ & Constant in $(0, 1/3)$. \\ \hline
  $\alpha$ & The bound on the arboricity of the dynamic graph $G$, i.e.\ at any moment the arboricity of $G$ is at most $\alpha$. \\ \hline
  $\Delta$ & The bound on the maximum degree of the dynamic graph $G$, i.e.\ at any moment the maximum degree of $G$ is at most $\Delta$. \\ \hline
  $N[v]$ & The closed neighborhood of $v$, i.e. $N[v] = \{v\} \cup N(v)$. \\ \hline
  $N[S]$ & The set of vertices dominated by $S$, i.e. $N[S] = \bigcup_{v \in S} N[v]$. \\ \hline
  $c(S)$ & The total cost of $S$, i.e. $c(S) = \sum_{v \in S} c_v$. \\ \hline
  $\OPT$ & The cost of an optimal dominating set. \\ \hline
  $L$ & The maximum level, $L = \lceil\log_{1+\epsilon}(Cn)\rceil + 1$. \\ \hline
  $\dom(v)$ & The dominating level of $v$. \\ \hline
  $\Dom(v)$ & The dominated level of $v$, $\Dom(v) = \max_{u \in N[v]} \dom(u)$. \\ \hline
  $\lev(v)$ & The level of $v$, with $\lev(v) \in [L]$. \\ \hline
  $\bs(v)$ & The base level of $v$, $\bs(v) = \lfloor\log_{1+\epsilon}(1/c_v)\rfloor$. \\ \hline
  $\bt(v)$ & The bottom level of $v$, $\bt(v) = \max_{u \in N[v]} \bs(u)$. \\ \hline
  $\tau_v$ & The minimum cost in $N[v]$, $\tau_v = \min_{u \in N[v]} c_u$. \\ \hline
  $x_v$ & The packing value of $v$, $x_v = \lambda(1+\epsilon)^{-\lev(v)}$. \\ \hline
  $\lambda$, $\xi$, $\gamma$ & Constants $\lambda = \frac{1}{(1+\epsilon)^2(3\alpha+1)}$, $\xi = \frac{\alpha+1}{(1+\epsilon)(3\alpha+1)}$, $\gamma = \frac{2\alpha+1}{3\alpha+1}$. \\ \hline
  $\wts(v)$ & The true weight of $v$, $\wts(v) = \sum_{u \in N[v]} x_u$. \\ \hline
  $\hat{\wts}(v)$ & The viewed weight of $v$, $\hat{\wts}(v) = \sum_{u \in N[v] \setminus F(v)} \viewx{v}{u}$. \\ \hline
  $\phi(v)$ & The dead weight of $v$. \\ \hline
  $\phi$ & The total dead weight, $\phi = \sum_{v \in V} \phi(v)$. \\ \hline
  $\delta_i$ & The deviation at level $i$, $\delta_i = \wat{i} \cdot |Q_i|$. \\ \hline
  $\delta$ & The total deviation, $\delta = \sum_{i=0}^{L} \delta_i$. \\ \hline
  $T$ & The set of tight vertices. \\ \hline
  $H$ & The set of heavy vertices, i.e. vertices with $\lev(v) = \bt(v)$. \\ \hline
  $H'$ & The set of cheapest neighbors chosen for the vertices in $H$. \\ \hline
  $A_i$ & The set of active vertices at level $i$. \\ \hline
  $P_i$ & The set of passive vertices with lazy level $i$. \\ \hline
  $V_i$ & The set of vertices with level $i$. \\ \hline
  $D_i$ & The set of vertices with dominating level $i$. \\ \hline
  $Q_i$ & The set of vertices that have an outdated copy at level $i$. \\ \hline
  $F(v)$ & The set of forgotten neighbors of $v$. \\ \hline
  $I(v)$ & The set of ignoring neighbors of $v$. \\ \hline
  $\zlev(v)$ & The lazy level of $v$, with $\zlev(v) \le \Dom(v)$. \\ \hline
  $\view{v}{u}$ & The viewed level of $u$ as stored by $v$. \\ \hline
  $\viewx{v}{u}$ & The viewed value of $u$ as stored by $v$, $\viewx{v}{u} = (1+\epsilon)^{-\view{v}{u}}$. \\ \hline
  $\hat{N}_i[v]$ & The list of vertices $u \in N[v] \setminus \Ign{v}$ that view $v$ at level $i$ (i.e.\ $\view{u}{v} = i$). \\ \hline
  Tight vertex & A vertex with $\hat{\wts}(v) + \phi(v) > \frac{\gamma c_v}{1+\epsilon}$. \\ \hline
  Slack vertex & A vertex that is not tight. \\ \hline
  Active vertex & A vertex with $\lev(v) = \Dom(v)$. \\ \hline
  Passive vertex & A vertex with $\Dom(v) < \lev(v)$. \\ \hline
\end{longtable}

\begin{longtable}{| l | p{13.5cm} |}
  \caption{Glossary of the notation relevant during the $\reset(k)$ subroutine.}\label{table:glossary-rebuild}\\
  \hline
  \textbf{Notation} & \textbf{Definition} \\ \hline
  \endfirsthead

  \hline
  \textbf{Notation} & \textbf{Definition} \\ \hline
  \endhead

  \hline
  \endfoot

  $\old$ & A superscript denoting values at the beginning of $\reset(k)$. \\ \hline
  $\mathcal{C}$ & The set of clean vertices, i.e.\ active vertices $v$ such that $\lev(v) \le k$ and passive vertices $v$ such that $\zlev(v) \le k$ and $\lev(v) \le k+1$. \\ \hline
  $\mathcal{D}$ & The set of dirty vertices, i.e.\ passive vertices $v$ such that $\zlev(v) \le k$ and $\lev(v) > k+1$. \\ \hline
  $\mathcal{D}'$ & The set of dirty vertices $v \in \mathcal{D}$ that were not filtered out during \cref{step:filter-dirty}. \\ \hline
  $F'(v)$ & The list of ``need-to-be-forgotten'' neighbors of $v$. \\ \hline
  $S$ & The set of vertices with $|F(v)| + |F'(v)| > (1+\epsilon)\alpha$. \\ \hline
  $T'$ & The set of vertices that were tight and at level at most $k$ before the rebuild, i.e.\ $T' = T^{\old}_{\le k}$. \\ \hline
  $f_v$ & The size of $F^{\old}(v)$, i.e. $f_v = |F^{\old}(v)|$. \\ \hline
  $f'_v$ & The size of $F'(v)$, i.e. $f'_v = |F'(v)|$. \\ \hline
\end{longtable}

\end{document}